\documentclass[twocolumn]{article}      

\usepackage{graphics} 

\usepackage{epsfig} 

\usepackage{mathptmx} 

\usepackage{times} 
 
\usepackage{amsmath} 

\usepackage{amsthm}

\usepackage{amssymb}  

\usepackage{color}

\usepackage{algorithm}

\usepackage{algorithmic}

\usepackage{enumerate}

\newtheorem{thm}{Theorem}

\newtheorem{dfn}{Definition}

\newtheorem{lm}{Lemma}

\newtheorem{prop}{Proposition}

\newtheorem{crl}{Corollary}

\newtheorem{rem}{Remark}

\newcommand{\RR}{{\mathbb R}}

\newcommand{\NN}{{\mathbb N}}

\newcommand{\norm}[1]{\lVert#1\rVert}

\newcommand{\dotex}{{\frac{d}{dt}}}

\title{\LARGE \bf
Intrinsic filtering on Lie groups	with applications to attitude estimation}

\author{Axel Barrau, Silv\`ere Bonnabel
\thanks{A. Barrau and S. Bonnabel are with MINES ParisTech, PSL Research University, Centre for robotics, 60 Bd St Michel 75006 Paris, France.
       {\tt\footnotesize [axel.barrau,silvere.bonnabel]@mines-paristech.fr}}%
}

\include{myLatex}

\graphicspath{{Figures/}}

\begin{document}

\maketitle

\thispagestyle{empty}

\pagestyle{empty}


\begin{abstract}
This paper proposes a probabilistic approach to the problem of intrinsic filtering of a system on a matrix Lie group with invariance properties. The problem of an invariant continuous-time model with  discrete-time measurements  is  cast into a rigorous stochastic and geometric framework. Building upon the theory of continuous-time invariant observers, we introduce a class of simple filters and study their properties (without addressing the optimal filtering problem). We show that, akin to the Kalman filter for linear systems, the error equation is a Markov chain that does not depend on the state estimate. Thus, when the filter's gains are held fixed, the noisy error's distribution is proved to converge to a stationary distribution, under some convergence properties of the filter with noise turned off. We also introduce two novel tools of engineering interest: the discrete-time Invariant Extended Kalman Filter, for which the trusted covariance matrix is shown to converge, and the Invariant Ensemble Kalman Filter. The methods are applied to attitude estimation, allowing to derive novel theoretical results in this field, and illustrated through simulations on synthetic data.
\end{abstract}

\section{Introduction}

The problem of probabilistic filtering aims at computing  the conditional expectation $\hat x(t)$ of the state $x(t)$ of a system driven by process and observation noises conditioned on the past inputs  and outputs. Filtering is a branch of stochastic process theory that has been strongly driven by applications in control and signal processing. In the particular case where the system is linear and corrupted by independent white Gaussian process and observation noises, the celebrated Kalman filter \cite{Kalman-1961}  solves the filtering problem. Yet, when the system is non-linear there is no general method to derive efficient filters,  and their design has encountered important difficulties in practice. Indeed, general filtering equations describing the evolution of the state conditioned on past outputs can also be derived, see e.g., \cite{pardoux} and references therein. However, the equations are not closed, and in most cases the conditional densities can not be computed by solving a finite dimensional equation.

In  engineering applications the most popular solution is the extended Kalman filter (EKF), which does not possess general theoretical properties, especially in a stochastic context,  and is sometimes prone to divergence. The EKF  amounts to linearize the system around the estimated trajectory, and build a Kalman filter for the linear model, which can in turn  be implemented on the non-linear model. Despite the fact that it is subject to several problems both theoretical and practical, one merit of the EKF  is to convey an estimation of the mean state $\hat x(t)$ together with an approximation of the covariance matrix $P(t)$ of the error, yielding a (trusted) ellipsoidal credible set that reflects the level of uncertainty in the estimation.  In order to capture more closely the distribution, several rather computationally extensive methods have recently attracted growing interest such as particle filters and unscented Kalman filters (UKF) \cite{UKF}.

In this paper we consider a general filtering problem on Lie groups in a stochastic setting, and a class of simple recursive finite dimensional filters, having a similar structure as the EKF, is introduced. Such  recursively defined filters allow straightforward digital computer implementation. Besides, as the Kalman filter, the corresponding algorithms will be readily implementable with very little understanding of the theory leading to their development, as illustrated by the concrete examples of sections \ref{asymptotic2vect}, \ref{Artificial_horizon} and \ref{sect::simus_gaussian}, which can be read independently of the rest of the paper. As concerns their theoretical properties, they are motivated by the three following desirable properties a simple recursive filter should ideally have: 1- forgetting of the initial condition, 2- being able to capture the extent of uncertainty in the estimate, 3- being efficient, that is, minimizing the error's dispersion within a given class. In the filtering problem, the estimation error is initially distributed according to a prior distribution  provided by the practitioner. The rationale for the first property is that, as the filter acquires information, the prior information should be given less (eventually no) importance by the filter. As the user generally possesses very little information on the system's state initially, the prior can be very far from the truth. Forgetting  the initial condition allows  the (subjective Bayesian) prior distribution to be as informative as noninformative without impacting the final estimates.  Moreover, it corresponds to convergence properties of the filter, and thus deals with the ability  to recover from catastrophic failures. In robotics, this leads to  robustness to erroneous beliefs due for instance to perceptual ambiguity.  The second property deals with the central motivation behind stochastic approaches. Indeed, all deterministic observers, no matter how strong properties they may have, have a common fundamental weakness: in the presence of sensor and model uncertainties, they provide no indication of the extent of uncertainty conveyed by the estimate. Being able to compute credible sets for the estimate may be needed in order to adapt control decisions to the situation, for instance to guarantee an absence of collision. The two first properties are general and apply to any class of sensible filters. The third property   addresses the issue of gain tuning. When working with a class of (suboptimal) filters, the gains can be tuned over a training set in order to minimize the level of uncertainty conveyed by the estimate. 

Several approaches to filtering for systems possessing a geometric
structure have been developed in previous literature. For stochastic
processes on Riemannian manifolds
\cite{ito1950stochastic,duncan1979stochastic} some results have been
derived, see e.g., \cite{ng1985nonlinear}. The specific situation where the
process evolves in a vector space but the observations belong to a manifold
has also been considered, see e.g.
\cite{pontier1987filtering,duncan1977some} and more recently
\cite{said2013filtering}. For systems on Lie groups  powerful tools  to
study the filtering equations (such as harmonic analysis
\cite{willsky-thesis,park2008kinematic,lo1986optimal,raey}) have been used,
notably in the case of bilinear systems  \cite{willsky1975estimation} and
estimation of the initial condition of a Brownian motion in
\cite{duncan1988estimation}.
An extension of
Gaussian distributions has also been developed and recently used for
Bayesian filtering on Lie groups in \cite{BayesianLieGroups2011} and also
in \cite{bourmaud2013discrete} which devises a general adaptation of the
EKF to Lie groups, without exploiting the geometrical properties of the
specific invariant case such as in the present paper. A somewhat different
but related approach to filtering consists of finding the path that best fits the data in a deterministic setting. It is thus related to optimal
control theory where geometric methods have long played an important role
\cite{brockett1973lie}. A certain class of least squares problems on the
Euclidian group has been tackled in \cite{han2001least}. This approach has
also been used to filter inertial data recently in \cite{sacconsecond}, \cite{mahony-et-al-IEEE}. 

Although the paper is concerned with a general theory on (matrix) Lie groups, the several derived filters will systematically be applied to a benchmark and motivating filtering example on the special orthogonal group $SO(3)$. The problem of attitude estimation consists of estimating the orientation of a rigid body,  which proves to be a central task in aeronautics, in order to stabilize or guide the system,  or in mobile robotics, where orientation is a part of the robot's localization. When the system is equipped with gyroscopes, the delivered signals can be integrated over time to compute an orientation increment. However, the initial orientation may not be known, and even if it is, the gyroscopes are subject to noise and drift, so that pure integration may quickly result in a very erroneous attitude. Thus the estimations are regularly checked  against some  vectors' orientation, such as the earth magnetic field measured by magnetometers, or the earth gravity vector measured by accelerometers under quasi-stationary flight assumptions. The explosion over the last decade of low-cost noisy sensors integrated in unmanned aerial vehicles (UAV's), and the limited computational power on-board,  have been a driving force for the development of simple filters on $SO(3)$. Work on the deterministic problem dates back to  \cite{salcudean1991globally,nijmeijer1999new} and has received much attention recently, see \cite{Vasconcelos,markley2005nonlinear,mahony-et-al-IEEE,bonnabel2008symmetry,MarSal2007a,sanyal2008global} to cite a few, and some related issues  have been addressed as well, like angular velocity estimation \cite{thienel2003coupled} or attitude control  \cite{salcudean1991globally,tayebi2007attitude}. The beauty of the intrinsic problem formulation on $SO(3)$ and the properties stemming from it have been exploited in various papers in the deterministic case, see e.g. \cite{mahony-et-al-IEEE,lageman2010gradient,arxiv-08}. The problem has also been cast into a minimum energy filtering framework in \cite{Zamni11}. In the stochastic framework, a general review of the classical techniques based on Gaussian assumptions is proposed in \cite{markley2005nonlinear}. A first attempt to systematically exploit the invariance properties to design stochastic filters on $SO(3)$ appears in the preliminary work \cite{barrau-bonnabel-cdc13}.

The contributions and the organization of the present paper are as follows. In Section \ref{Problem_setting},  the problem of stochastic filtering on Lie groups with a continuous-time noisy model and discrete-time observations is rigorously posed. From a theoretical viewpoint, the choice of discrete-time measurements allows to circumvent some technical difficulties and to cast the problem into a perfectly rigorous framework. As this work is mainly guided by applications to aerospace, this choice  is relevant as the evolution equations of the orientation generally rely on the integration of noisy inertial measurements that are delivered at a high rate, whereas the observations may be delivered at low(er) rate (GPS, vision). In order to derive proper filters, the problem is then transformed into a complete discrete-time model. The discretization is not straightforward because unlike the linear case, difficulties arise from non-commutativity.

Guided by the theory of symmetry-preserving observers that has been exclusively concerned with the continuous-time and deterministic case so far, we propose in Section \ref{InvariantFiltering} a class of natural filters. They appear as a mere transposition of linear observers (such as the linear Kalman filter) to the group, but where the addition has been replaced with the group multiplication law (due to non-commutativity, there is always some freedom in the transposition as left and right multiplications are different operations). We prove the proposed filters retain the invariances of the problem, and that the discrete-time error's evolution is independent of the system's trajectory, inheriting the properties of the deterministic continuous-time case \cite{arxiv-08} (or merely the linear case). Thus, the gains can be computed off-line, and the proposed filters require very few on-board computational resources. This is one of their main advantages over current filters for the attitude estimation problems considered in simulations.


The proposed class is very broad and the tuning issue is far from trivial. Sections \ref{Fixed_gains} and \ref{Gaussian_filters} explore two different routes, depending on if we are interested in the asymptotic or transitory phase. In Section \ref{Fixed_gains}, we propose to hold the gains fixed over time. As a result, the error  equation becomes a \emph{homogeneous} Markov chain. We first prove some new results about global convergence in a deterministic and discrete-time framework. Then, building upon the homogeneity property of the error, we prove that if the filter with noise turned off admits almost global convergence properties, the error with noise turned on converges to a stationary distribution. Mathematically this is a very strong and to some extent surprising result. From a practical viewpoint, the gains can be tuned numerically to minimize the asymptotic error's dispersion. This allows to ``learn" sensible gains for very general types of noises. The theory is applied to two examples and gives convergence guarantees in each case. First an attitude estimation problem using two vector measurements and a gyroscope having isotropic noise (Section \ref{asymptotic2vect}), then the construction of an artificial horizon with optimal gains, for a non-Gaussian noise model (Section \ref{Artificial_horizon}). Each application is a contribution in itself and can be implemented without reading the whole paper.

In Section \ref{Gaussian_filters}, we propose to optimize the convergence during the transitory phase using Gaussian approximations. We first introduce a method of engineering interest: the Invariant Extended Kalman Filter (IEKF) in discrete time. The IEKF was introduced in continuous time in \cite{bonnabel2007left} and developped and applied to localization problems in \cite{bonnabel2009invariant} and very recently in \cite{barczyk2013invariant}. A discrete time version on $SO(3)$ was proposed in the preliminary paper \cite{barrau-bonnabel-cdc13}. The idea consists of linearizing the invariant error for the class of filters introduced in Section \ref{InvariantFiltering} around a fixed point (the identity element of the matrix group), build  a Kalman filter for the linear model and implement it on the non-linear model. As the linearizations always occur around the same point, the linearized model is time-invariant and thus the Kalman gains, as well as the Kalman covariance matrix, are proved to converge to fixed values. When on-board storage and computational resources are very limited, this advantageously allows   to replace the gain with its asymptotic value. The IEKF is compared to the well-known MEKF \cite{lefferts1982kalman,Crassidis} and UKF \cite{UKF,crassidis2003unscented} on the attitude estimation problem, and simulations illustrate some convergence properties that the latter lack. In the case where the error equation is totally autonomous, we introduce a new method based on off-line simulations, the IEnKF, which outperforms the other filters in case of large noises by capturing very accurately the error's dispersion.

\section{Problem setting}
\label{Problem_setting}

\newtheorem{theorem}{Proposition}

\subsection{Considered continuous-time model}
Consider a state variable $\chi_t$ taking values in a matrix Lie group $G$ with neutral element $I_d$, and the following continuous-time model with discrete measurements:
\begin{align}
\dotex{\chi}_t & = (\upsilon_t+w_t) \chi_t + \chi_t \omega_t \label{dynamic}\\
Y_n & =h(\chi_{t_n},V_n)\label{eq2}
\end{align}
where $\upsilon_t$ and $\omega_t$ are inputs taking their values in the Lie algebra $\mathfrak{g}$ (i.e. the tangent space to the identity element of $G$), $w_t$ is a continuous white Langevin noise with diffusion matrix $\Sigma_t$, whose precise definition will be discussed below in Subsection  \ref{app::Strato}.   $(Y_n)_{n \geqslant 0}$ are the discrete-time observations, belonging to some measurable space $\mathcal{Y}$ and $V_n$ is a noise taking values in $\RR^p$ for an integer $p>0$. We further make the following additional assumption:
\paragraph*{Assumption 1 (left-right equivariance)}
The output map $h$ is left-right equivariant, i.e. there exists a left action of $G$ on $\mathcal{Y}$ such that we have the equality in law:
\begin{align}
 \forall g,\chi \in G, &  ~h(\chi  g,V_n) \overset{\mathcal{L}}{=}g^{-1}h(\chi,V_n) \\
 \text{ and } & h(\chi  g,0) = g^{-1}h(\chi,0) 
\end{align}

The reader who is not familiar with stochastic calculus and Lie groups can view $\chi_t$ as a rotation matrix, and replace the latter property with the following more restrictive assumption:
\paragraph*{Assumption 1 bis}
For any $n \geqslant 0$ there exists a vector $b$ such that $h(\chi,V_n) {=} \chi^{-1}(b+V_n)$.

The assumptions could seem restrictive but are verified in practice in various cases as shown by the following examples, which will provide the reader with a more concrete picture than the formalism of Lie groups.


\subsection{Examples}

\subsubsection{Attitude estimation on flat earth}
\label{expl::2vect}
Our motivating example for the model \eqref{dynamic}-\eqref{eq2} is the attitude estimation of a rigid body assuming the earth is flat, and  observing  two vectors:

\begin{equation}
\label{Ex1}
  \begin{aligned}
\dotex{R}_t & = R_t (\omega_t+w_t^\omega)_{\times}, \\
Y_n & = (R_{t_n}^T b_1+V_n^1, R_{t_n}^T b_2 + V_n^2), 
  \end{aligned}
\end{equation}
where as usual $R_t\in{SO(3)}$ represents the rotation that maps the body frame to the earth-fixed frame, and where $\omega_t\in\RR^3$ is the instantaneous rotation vector, and $w_t^\omega \in\RR^3$ is a continuous Gaussian white noise representing the gyroscopes' noise. We have let $(x)_\times\in\RR^{3\times 3}$ denote the skew matrix associated with the cross product with a three dimensional vector, i.e., for $a,x\in\RR^3$, we have $(x)_\times a=x\times a$. $(Y_n)_{n\geq 0}$ is a sequence of discrete noisy measurements of two vectors $b_1,b_2$ of the earth fixed frame verifying $b_1 \times b_2 \not= 0$, and  $V_n^1$ and $V_n^2$ are sequences of independent isotropic Gaussian white noises. Note that the noise $w_t^\omega$ is defined, as the input $\omega_t$, in the body frame (in other words it is multiplied on the left). Thus equations \eqref{Ex1} do not match \eqref{dynamic}-\eqref{eq2} which correspond to a noise defined in the earth-fixed frame. This can be remedied in the particular case where the gyroscope noise is isotropic, a restrictive yet relevant assumption in practice.
\begin{dfn}
A Langevin noise $w_t$ of $\frak{g}$ is said isotropic if $\tilde w_t:=g w_t g^{-1} \overset{\mathcal{L}}{=} w_t$ for any $g \in G$.
\end{dfn}
Note that for a noise taking values in $\mbox{so}(3)$ this definition corresponds to the physical intuition of an isotropic noise of $\RR^3$. With this additional assumption the equation can be rewritten:
\begin{equation*}
  \begin{aligned}
\dotex{R}_t & = (\tilde w_t^\omega)_{\times} R_t + R_t (\omega_t)_{\times}, \\
Y_n & = (R_{t_n}^T b_1+V_n^1, R_{t_n}^T b_2 + V_n^2),
  \end{aligned}
\end{equation*}
which corresponds to \eqref{dynamic}-\eqref{eq2}.

\begin{rem}
If the gyrometer noise is not isotropic the new noise $\tilde{w}_t^{\omega}$ is related to $R_t$ by $\tilde{w}_t^{\omega}=R_t w_t^{\omega} R_t^T$. Depending on the degree of anisotropy this can prevent the use of methods based on the autonomy of the error (see \ref{Fixed_gains}) but not of the discrete-time IEKF (see \ref{simus}).
\end{rem}

\subsubsection{Attitude estimation on a round rotating Earth}
\label{expl::1vect}
Another interesting case is  attitude estimation on ${SO}(3)$ using an observation of the vertical direction (given by an accelerometer) and taking into account the rotation of the earth. Whereas the flat earth assumption perfectly suits low-cost gyroscopes, precise gyroscopes can measure the complete attitude by taking into account the earth's rotation. We define a geographic frame with axis North-west-up. The attitude $R_t$ is the transition matrix from the body reference to the geographic reference. The earth instantaneous rotation vector $\upsilon \in \RR^3$ and the gravity vector $g \in \RR^3$ are expressed in the geographic reference. Both are constant in this frame. The gyroscope gives a continuous rotation speed $\omega_t$, disturbed by an isotropic continuous white noise $w_t^\omega$. The equation of the considered system reads:
\begin{equation*}
  \begin{aligned}
\dotex{R}_t & = (\upsilon)_\times  R_t + R_t (\omega_t + w_t^\omega)_\times, \\
Y_n & = R_{t_n}^T (g + V_n) 
  \end{aligned}
\end{equation*}
As $w_t^\omega$ is supposed isotropic the equation can be rewritten:
\begin{equation*}
  \begin{aligned}
\dotex{R}_t & = (\upsilon+\tilde w_t^\omega)_\times R_t + R_t (\omega_t)_\times, \\
Y_n & = R_{t_n}^T (g + V_n),
  \end{aligned}
\end{equation*}
which corresponds to \eqref{dynamic}-\eqref{eq2}.


\subsubsection{Attitude and velocity estimation}
We give here an example on a larger group. Consider the attitude $R_t \in SO(3)$ and speed $v_t \in \RR^3$ of an aircraft evolving on a flat earth, equipped with a gyroscope and an accelerometer. The gyroscopes gives continuous increments $\psi_t \in \RR^3$ with isotropic noise $w_t^\psi \in \RR^3$, and the accelerometers gives continuous increments $a_t \in \RR^3 $ with isotropic noise $w^a_t \in \RR^3$. The aircraft has noisy speed measurements in the earth reference frame $Y_n = v_{t_n}+V_n$, where $V_n$ is a supposed to be a Gaussian isotropic white noise. The equations read:
\begin{equation*}
  \begin{aligned}
\dotex{R}_t & = R_t (\psi_t + w_t^\psi)_\times, \\
\dotex{v}_t & = R_t (a_t+w_t^a) + g,
  \end{aligned}
\end{equation*}
where $R_t$ is the rotation mapping the body frame to the earth-fixed frame, $v_t$ is the velocity expressed in the earth fixed frame, and $g$ is the earth gravity vector, supposed to be (locally) constant.
Using the matrix Lie group $SE(3)$ we introduce:
\[
A_t=\left( \begin{array}{cc}
R_t^T & R_t^T v_t \\
0 & 1 \end{array} \right) ~ ; ~
\omega_t=\left( \begin{array}{cc}
0 & g \\
0 & 0 \end{array} \right) 
\]
\[
\upsilon_t=\left( \begin{array}{cc}
-(\psi_t)_\times & a_t \\
0 & 0 \end{array} \right) ~ ; ~ 
w_t=\left( \begin{array}{cc}
(w_t^\psi)_\times & w_t^a \\
0 & 0 \end{array} \right)
\]
The problem can then be rewritten under the form \eqref{dynamic}-\eqref{eq2}:
\begin{equation*}
  \begin{aligned}
\dotex{A}_t & = (\upsilon_t+w_t) A_t + A_t \omega_t, \\
Y_{t_n} & = A_{t_n}^{-1}  \left( \begin{array}{c}
V_n \\
1 \end{array} \right),
  \end{aligned}
\end{equation*}


\subsection{Interpretation of Langevin noises on Lie Groups}
\label{app::Strato}

This present subsection provides some mathematical considerations about white noises on Lie groups. It can be skipped  by the uninterested reader who is directly referred to the model \eqref{discretization}.

Equation \eqref{dynamic} is actually a Langevin equation that suffers from poly-interpretability because of its non-linearity, and its meaning  must be clarified in the rigorous framework of stochastic calculus. Stratonovich stochastic differential equations on a Lie group can be intrinsically defined as in e.g.  \cite{LevyLieGroups}. A somewhat simpler (but equivalent) approach consists of using the natural embedding of the matrix Lie group $G$ in a matrix space and to understand the equation in the sense of Stratonovich, see e.g. \cite{said2009estimation}. The mathematical reasons stem from the fact that the resulting stochastic process is well-defined  on the Lie group, whereas this is not the case when opting for an Ito interpretation as underlined by the following easily provable proposition.
\begin{prop}
If the stochastic differential equation \eqref{dynamic} is taken in the sense of It\^ o, for $G=SO(3)$ embedded in $\RR^{3\times 3}$, an application of the It\^o formula to $\chi_t^T \chi_t$ shows that the solution  almost surely leaves the submanifold $G$ for \emph{any} time $t>0$.
\end{prop}
Besides, the physical reasons for this stem from the fact that the sensors' noise are never completely white, and for colored noise the Stratonovich interpretation provides a better approximation to the true solution than Ito's, as advocated by the  result of Wong and Zakai \cite{wong1969riemann}.


\subsection{Exact discretization of the considered model}
\label{section::discretization}
To treat rigorously the problem of integrating discrete measurements we need to discretize the continuous model with the same time step as the measurements'.
Unlike the general case of non-linear estimation,  the  exact discrete-time  dynamics corresponding to Equation \eqref{dynamic} can be obtained, as proved by the following  result:
\begin{thm}
\label{discrete_form}
Let $\chi_n=\chi_{t_n}$. Then the discrete system $(\chi_n,Y_n)$ satisfies the following equations:
 \begin{equation}
 \label{discretization}
  \begin{aligned}
  \chi_{n+1} & = \Upsilon_n W_n \chi_n \Omega_n, \\
  Y_n & = h(\chi_n,V_n),
   \end{aligned}
 \end{equation}
where $W_n$ is a random variable with values in $G$ and whose law depends on the values taken by $\upsilon_t$ for $t \in [t_n, t_{n+1}]$ and on the law of $w_t$ for $t \in [t_n, t_{n+1}]$, and $\Upsilon_n$ and $\Omega_n$ are elements of $G$ which only depend on the values taken respectively by $\upsilon_t$ and $\omega_t$ for $t \in [t_n, t_{n+1}]$.
\end{thm}
\begin{proof}
For $n \in \mathbb{N}$ consider the value of the process $\chi_t$ is known until time $t_n$. Let $\Upsilon_t$ and $\Omega_t$ be the solutions of the following equations:
\begin{align*}
 & \Upsilon_{t_n}  = I_d, ~ \frac{d}{dt} \Upsilon_t  = \upsilon_t \Upsilon_t \qquad \\
 \text{and} ~ & \Omega_{t_n} = I_d, ~  
  \frac{d}{dt} \Omega_t = \Omega_t \omega_t
\end{align*}
Let $W_t$ be the solution of the following (Stratonovich) stochastic differential equation: 
\begin{align*}
W_{t_n} = I_d, \hspace{1 cm} \frac{d}{dt} W_t = \Upsilon_t^{-1 } w_t \Upsilon_t W_t
\end{align*}
Note that $\Upsilon_t^{-1 } w_t \Upsilon_t$ being in $\frak{g}$, $W_t$ is ensured to stay in $G$. Define the process $\chi_{t|t_n}=\Upsilon_t W_t \chi_{t_n} \Omega_t$. We will show that for $t>t_n$ the processes $\chi_t$ and $\chi_{t|t_n}$ verify the same stochastic differential equation. Indeed:
\begin{align*}
\frac{d}{dt} \chi_{t|t_n} & = (\frac{d}{dt} \Upsilon_t) W_t \chi_t \Omega_t + \Upsilon_t (\frac{d}{dt}W_t) \chi_t \Omega_t + \Upsilon_t W_t \chi_t (\frac{d}{dt}\Omega) \\
 & = \upsilon_t \Upsilon_t W_t \chi_t \Omega_t + w_t \Upsilon_t W_t \chi_t \Omega_t + \Upsilon_t W_t \chi_t \Omega_t \omega_t \\
 & = (\upsilon_t+w_t)\chi_{t|t_n} +\chi_{t|t_n} \omega_t
\end{align*}
Thus the two processes have the same law at time $t_{n+1}$, i.e. $\chi_{n+1}$ and $W_{t_{n+1}} \Upsilon_{t_{n+1}} \chi_{t_n}\Omega_{t_{n+1}}$ have the same law. Letting $\Upsilon_n=\Upsilon_{t_{n+1}}$, $\Omega_n=\Omega_{t_{n+1}}$ and $W_n=W_{t_{n+1}}$ we obtain the result. 
\end{proof}
\begin{rem}
In many practical situations (for instance examples \ref{expl::2vect} and \ref{expl::1vect}), the Langevin noise $w_t$ is isotropic and we have thus $\Upsilon_t^{-1}w_t\Upsilon_t \overset{\mathcal{L}}{=} w_t$. Note that, the variable $W_t$ depends also only on the law of $w_t$ for $t \in [t_n, t_{n+1}]$.
\end{rem}
In the sequel, for mathematical reasons (the equations do not suffer from poly-interpretability), tutorial reasons (the framework of diffusion processes on Lie groups needs not be known), and practical reasons (any filter must be implemented in discrete time),  we will systematically consider the discrete-time model \eqref{discretization}. Moreover, the noise $W_n$ will be a general random variable in $G$, not necessary a solution of the stochastic differential equation $\frac{d}{dt}W_t=\Upsilon_t^{-1} w_t \Upsilon_t W_t$.


\section{A class of discrete-time intrinsic filters}
\label{InvariantFiltering}


\subsection{Preliminary: Kalman filter for linear stationary systems with drift }

\label{linear_case}

Consider the particular case of a linear system in $\RR^N$ (with arbitrary $N$) of the form $\frac{d}{dt} X_t = b_t + w_t$, $w_t$ being a white noise and $b_t$ a known deterministic input. Assume we have discrete measurements of the form $Y_n = H X_{t_n} + V_n \in \RR^p$, $V_n$ being a Gaussian noise in $\RR^p$. A straightforward discretization yields:
\begin{equation}
\label{separation_obs}
\begin{aligned}
X_{n+1} & = X_n + B_n+ W_n, \\
Y_n & = H X_n + V_n,
\end{aligned}
\end{equation}
With $B_n=\int_{t_n}^{t_{n+1}} b_s ds$ and $W_n=\int_{t_n}^{t_{n+1}} w_s ds$ is the value of a Brownian motion in $\RR^N$. The optimal filter in this case is the celebrated linear Kalman filter whose equations are:
\begin{equation}
\label{linear_estimator}
\begin{aligned}
\hat{X}'_{n+1} & =\hat{X}_n+B_n, \\
\hat{X}_{n+1}  & =\hat{X}'_{n+1}+K_{n+1}(Y_{n+1} - H\hat{X}'_{n+1})
\end{aligned}
\end{equation}
The first equation is the so-called prediction step, and the second one the correction step. 
Let $e'_n=\hat{X}'_n-X_n$ and $e_n=\hat{X}_n-X_n$ denote the predicted and corrected state estimation errors, we get:
\begin{equation}
\begin{aligned}
e'_{n+1} & = e_n-W_n, \\
e_{n+1} & = e'_{n+1}-K_{n+1}(H e'_{n+1}-V_{n+1})
\end{aligned}
\end{equation}
The error equation is autonomous and can be optimized independently from the dynamical inputs $B_n$'s. As a consequence, the optimal gains $K_n$ can be computed off-line. It is a well-known result that they (as well as the error distribution) converge under observability and detectability  assumptions, so all computations can be done off-line and  only the first values need be stored, in which case only the equations \eqref{linear_estimator} must be computed on-line, requiring extremely restricted computational power on board. Mimicking the form of equations  \eqref{linear_estimator} on the Lie group $G$ where the addition is naturally replaced with the group multiplication law, leads to a similar result, as shown in the next section.


\subsection{Proposed intrinsic filters}
\label{Filter_definition}
Inspired by the linear case and the theory of continuous-time  symmetry-preserving observers on Lie groups \cite{arxiv-08} we propose to mimic the prediction and update steps \eqref{linear_estimator} replacing the addition with the group multiplication yielding the class of filters defined by:
\begin{align}
\hat{\chi}'_{n+1} & = \Upsilon_n \hat{\chi}_n \Omega_n, \label{non_linear_estimator1} \\
\hat{\chi}_{n+1} & = K_{n+1}(\hat{\chi'}_{n+1}Y_{n+1}) \hat{\chi'}_{n+1}, \label{non_linear_estimator2}
\end{align}
where $K_{n+1}(\bullet)$ can be any function of $\mathcal{Y} \rightarrow G$,  ensuring $K(h(I_d,0))=I_d$  . There are two ways to understand the links between \eqref{linear_estimator} and \eqref{non_linear_estimator1}-\eqref{non_linear_estimator2}. First, by defining a an estimation error on the group $G$ by $\chi_{n+1} \hat{\chi}_{n+1}^{-1}$ (see below), and by using the left-right equivariance hypothesis, which allows to interpret $\hat{\chi'}_{n+1}Y_{n+1}=h(\chi_{n+1} \hat{\chi'}_{n+1}^{-1}, V_{n+1})$  as a group equivalent to $H(X_{n+1}' - \hat{X}_{n+1})+V_{n+1}$. Secondly, by viewing the linear case of Section \ref{linear_case} as a specific case of the proposed approach through the following proposition whose proof has been moved to the Appendix. 

\begin{prop}
\label{isomorphism}
There exists an isomorphic representation of $\RR^N$ as a matrix Lie group such that \eqref{separation_obs} takes the canonical form \eqref{dynamic}-\eqref{eq2}, and \eqref{linear_estimator} becomes the invariant filter \eqref{non_linear_estimator1}-\eqref{non_linear_estimator2}.
\end{prop}



We now define the invariant output errors, which are a transposition of the linear error to our multiplicative group:
\begin{align}
\eta_n =\chi_n\hat{\chi}_n^{-1}, \hspace{1 cm} \eta'_n =  \chi_n \hat{\chi'}_n^{-1}
\end{align}
We have the following striking property, that is similar to the linear case:
\begin{thm}
The error variables $\eta_n$ and $\eta'_n$ are Markov processes, and are independent of the inputs $\Omega_n$'s.
\end{thm}

\begin{proof}

The equations followed by $\eta_n'$ and $\eta_n$ read:
\begin{equation} \label{Markov1} \eta'_{n+1} =\chi_{n+1}\hat{\chi'}_{n+1}^{-1} 
= \Upsilon_n W_n \chi_n \Omega_n \Omega_n^{-1} \hat{\chi}_n^{-1}\Upsilon_n^{-1}  = \Upsilon_n W_n \eta_n \Upsilon_n^{-1} \end{equation} and:
\begin{equation}
\label{Markov2}
\begin{aligned} \eta_{n+1} & = \chi_{n+1} \hat{\chi'}_{n+1}^{-1} K_{n+1}(\hat{\chi'}_{n+1}Y_{n+1})^{-1}  \\
 & = \eta'_{n+1} K_{n+1}(\hat{\chi'}_{n+1}h({\chi}_{n+1},V_{n+1}))^{-1} \\
 & \overset{\mathcal{L}}{=}  \eta'_{n+1} K_{n+1}(h(\eta'_{n+1},V_{n+1}))^{-1},
\end{aligned}
\end{equation}
thanks to the equivariance property of the output.
\end{proof}

The most important consequence of this property is that if the inputs $\upsilon_t$ are known in advance, or are fixed, as it is the case in Examples \ref{expl::2vect} and \ref{expl::1vect}, the gain functions $K_n$ can be optimized off-line, independently of the trajectory followed by the system. In any case, numerous choices are possible to tune the gains $K_n$ and the remaining  sections are all devoted to various  types of methods to tackle this problem.

\section{Fixed gains filters}
\label{Fixed_gains}

In certain cases, one can build an (almost) globally convergent observer for the associated deterministic system, i.e., with noise turned off, by using a family of  constant gain function $K_n( \bullet ) \equiv K( \bullet )$. If the filter with noise turned off has the desirable  property of forgetting its initial condition, convergence to a single point is impossible to retain, because of the unpredictability of the noises that ``steer" the system, but convergence of the distribution can be expected, assuming as in the linear case that $\Upsilon_n$ is independent of $n$. Indeed in this section we prove that, when noise is turned on, the error forgets its initial distribution under mild conditions.  The results are illustrated by an attitude estimation example for which we propose an intrinsic filter having strong convergence properties. 

It should be noted that in practice, the convergence to an invariant distribution  allows in turn to pick the most desirable gain $K( \bullet )$ among the family based on a performance criterion, such as   convergence speed, or filter's precision (that is, ensuring low error's dispersion). This fact will be illustrated by an artificial horizon example of Subsection \ref{Artificial_horizon}.

\subsection{Convergence results}
\label{cv_results}
Here the left-hand inputs $\Upsilon_n$ are assumed fixed. This, together with constant gains $K_n$, makes the error sequence $\eta_n$ a homogeneous Markov chain. Thus, under appropriate technical conditions, the chain has a unique stationary distribution and the sequence converges in distribution to this invariant
distribution. Let $d$ denote a right-invariant distance on the group $G$. We propose the following assumptions:
\begin{enumerate}
\item Confinement of the error: there exists a compact set $C$ such that $\forall n\in \mathbb{N}, \eta_n \in C $ a.s. for any $\eta_0\in C$\label{prop::compact}
\item Diffusivity: the process noise has a continuous part with respect to Haar measure, with density positive and uniformly bounded from zero in a ball of radius $\alpha>0$ around $I_d$. \label{prop::diff}
\item Reasonable output noise: $\forall g \in G$,  $\mathbb{P}[g K(h(g,V_n))^{-1} \in \mathcal{B}_o(g K(h(g,0))^{-1},\frac{\alpha}{2})] > \epsilon'$ for some $\epsilon'>0$. \label{prop::output}
%
%
\end{enumerate}
The second assumption implies, and can in fact be replaced with, the more general technical assumption that there exists $\epsilon >0$ such that for any subset $U$ of the ball $ \mathcal{B}_o  (I_d,\alpha)$ we have $P(W_n \subset U) > \epsilon \mu(U)$ for all $n\geq 0$ where $\mu$ denotes the Haar measure. Those noise properties are relatively painless to verify, whereas the confinement property although stronger is automatically verified whenever $G$ is compact, e.g. $G=SO(3)$. 
Intuitively,  the last two assumptions guarantee the  error process is well approximated by its dynamics with noise turned off, followed by a small diffusion. In the theory of Harris chains, the latter diffusion step is a key element to allow probability laws to mix at each step and eventually forget their initial distribution.  
\begin{thm}
 \label{thm::global}
For constant left-hand inputs $\Upsilon_n \equiv \Upsilon$, consider the filter:
\begin{align}
\hat{\chi}'_{n+1} & = \Upsilon \hat{\chi}_n \Omega_n \\
\hat{\chi}_{n+1} & = K(\hat{\chi'}_{n+1}Y_{n+1}) \hat{\chi'}_{n+1}
\end{align}Suppose that Assumptions 1)-2)-3) are verified, where the compact set satisfies $C=cl(C^o)$, $cl$ denoting the closure and $^o$ the interior. When noises are turned off, the error equation \eqref{Markov1}-\eqref{Markov2} becomes:
\begin{equation}
\label{noiseless}
\begin{aligned}
\gamma'_{n+1} & = \Upsilon \gamma_n \Upsilon^{-1}, \\
\gamma_{n+1} &  {=}  \gamma'_{n+1} K_{n+1}(h(\gamma'_{n+1},0))^{-1} ,
\end{aligned}
\end{equation}
Suppose that for any $\gamma_0 \in C$, except on a set of null Haar measure, $\gamma_n$ converges to $I_d$. Then there exists a unique stationary distribution $\pi$ on $G$ such that for any prior law $\mu_0$ of the error $\eta_0$ supported by $C$, the law $(\mu_n)_{n\geq 0}$ of $(\eta_n)_{n\geq 0}$ satisfies the total variation (T.V.) norm convergence property: $$\lim_n~\norm{\mu_n-\pi}_{T.V.} \to 0$$
\end{thm}
 \begin{crl}
 \label{thm::global2}
When the group $G$ is compact, The convergence results of Theorem \ref{thm::global} hold globally, i.e. without the confinement assumption 1).
\end{crl}

\begin{thm}
 \label{thm::subgroup}
Under the assumptions of Theorem \ref{thm::global}, assuming only  $h(\gamma_n,0) \rightarrow 0$ instead of almost global convergence of $(\gamma_n)_{n\geq 0}$, that $G$ is compact, the set $K=\{g \in G, h(g,0)=h(I_d,0) \}$ connected and $h(\Upsilon,0)=h(I_d,0)$, the results of Theorem \ref{thm::global} are still valid. Moreover, if $W_n$ is isotropic, we have $\pi(\tilde{\Upsilon} V)=\pi(V)$ for any $\tilde{\Upsilon} \in K$ commuting with $\Upsilon$ ($\tilde{\Upsilon} \Upsilon = \Upsilon \tilde{\Upsilon}$).
\end{thm}

The proofs of the results above have all been moved to the Appendix.

\subsection{Application to attitude estimation}
\label{asymptotic2vect}

Consider the attitude estimation example of Subsection \ref{expl::2vect}. In a deterministic and continuous time setting, almost globally converging observers have been proposed in several papers (see e.g. \cite {salcudean1991globally,mahony-et-al-IEEE}) and have since been analyzed and extended in a number of papers. In order to apply the previously developed theory to this example, the challenge is twofold. First, the deterministic observer must be adapted to the discrete time and proved to be almost globally convergent. Then, the corresponding filter must be proved to satisfy the assumptions of the theorems above in the presence of noise.
In discrete time, the system equations read:
\begin{equation*}
\begin{aligned}
R_{n+1}  = & W_n R_n \Omega_n, \\
Y_n  = & (R_{n}^T b_1+V_n^1, R_{n}^T b_2 + V_n^2)
\end{aligned}
\end{equation*}
with the notations introduced in \ref{expl::2vect}. We propose the following filter on $SO(3)$:
\begin{equation}\label{SOOO:eq}
\begin{aligned}
\hat{R}'_{n+1} & = \hat{R}_n \Omega_n, \\
\hat{R}_{n+1} & = K(\hat{R'}_{n+1} Y_{n+1}) \hat{R'}_{n+1},\\\text{with }
K(y_1,y_2)&=\exp( k_1 (y_1\times b_1) + k_2 (y_2 \times b_2)), \\
 & k_1>0,~  k_2>0,~ k_1+k_2 \leqslant 1\end{aligned}\end{equation}
\begin{prop}
\label{prop::K1}
With noise turned off, the discrete invariant observer \eqref{SOOO:eq} is almost globally convergent, that is, the error converges to $I_d$ for any initial condition except one.
\end{prop}
The proof has been moved to the Appendix, only the main idea is given here. For the continuous-time deterministic problem it is known, and easily seen, that $E: \gamma \rightarrow k_1 ||\gamma^T b_1-b_1||^2 + k_2 ||\gamma^T b_2 - b_2||^2$ is a Lyapunov function, allowing to prove almost global convergence of the corresponding observer. In the discrete time deterministic case, the function above remains a Lyapunov function for the sequence $(\gamma_n)_{n\geq 0}$, allowing to derive Proposition \ref{prop::K1}. This is not trivial to prove, and stems from the more general following novel result:
\begin{prop}
\label{thm::discrete}
Consider a Lie group $G$ equipped with a left-invariant metric  $\langle .,. \rangle$, and a left-invariant deterministic discrete equation on $G$  of the form:
\begin{align*}
\gamma_{n+1}=\gamma_n \exp(-k(\gamma_n))
\end{align*}
Assume there exists a  $C^2$  function $E: G \rightarrow \RR_{\geq 0}$  with bounded sublevel sets, a global minimum at $I_d$, and a continuous and strictly positive function $u: G \rightarrow \RR_{>0}$ such that: $\forall x \in G, k(x) = u(x) [x^{-1}.grad_E(x)] $. If the condition $\forall x \in G, |\frac{\partial k}{\partial x}| \leqslant 1$ (for the operator norm) is verified, for any initial value $\gamma_0$ such that $I_d$ is the only critical point of $E$ in the sublevel set $\{x\in G\mid E(x)\leq E(\gamma_0)\}$ we have:
\[
\gamma_n \underset{n \rightarrow \infty}{\rightarrow} I_d
\]
\end{prop}
The proof has been moved to the Appendix.
Note that, the latter  property is closely related to Lemma 2 of \cite{tron2012riemannian}. 
Using Theorem \ref{thm::global} and Proposition \ref{prop::K1} we finally directly get:
\begin{thm}
The distribution of the error variable of the invariant filter \eqref{SOOO:eq} converges for the T.V. norm to an asymptotic distribution, which does not depend on its initial distribution.
\end{thm}

\subsection{Learning robust gains: application to the design of an artificial horizon}
\label{Artificial_horizon}

As proved in \ref{cv_results}, under appropriate conditions the error variable is a converging Markov chain whose asymptotic law depends on the gain function but not on the trajectory followed by the system (which is a major difference with most nonlinear filters, such as the EKF). Hence a fixed gain can be asymptotically optimized off-line, leading to a very low numerical cost of the on-line update.

A classical aerospace problem is the design of an artificial horizon using an inertial measurement unit (IMU). An estimation of the vertical is maintained using the observations of the accelerometer (which senses the body acceleration minus the gravity vector, expressed in the body frame) and the stationary flight assumption according to which the body's linear velocity is constant. The problem is that this approximation is not valid in dynamical phases (take-off, landing, atmospheric turbulence), which are precisely when the  artificial horizon is most needed. The problem is generally stated as follows:
\begin{equation*}
\begin{aligned}
\dotex{R_t} & =R_t (\omega_t + w_t)_{\times}, \\
Y_n & =R_{t_n}^Tg +V_n+N_n,
\end{aligned}
\end{equation*}
where $R_t$ is the attitude of the aircraft (the rotation from body-frame coordinates to inertial coordinates), $\omega_t$ is the continuous-time gyroscope increment and $Y_n$ is the observation of the accelerometer. The sensor noises $w_t$ and $V_n$ can be considered as Gaussian, and $N_n$ represents fluctuations due to accelerations of the aircraft that we propose to model as follows: $N_n$ is null with high probability but when non-zero it can take large values. The $N_n$'s are assumed to be independent as usually. 

\subsubsection{Convergence results}Consider the following class of filters:
\begin{equation}\label{horizon}
\begin{aligned}
\hat{R}'_{n+1} & = \hat{R}_n \Omega_n, \\
\hat{R}_{n+1}  & = K(\hat{R'}_{n+1} Y_{n+1}) \hat{R'}_{n+1}, \\
\text{with } K(y)  & = \exp(f_{k,\lambda}(y))
\end{aligned}
\end{equation}
where $f_{k,\lambda}(x) = k.\min(\text{angle}(x,g), \lambda)  \frac{x \times g}{||x \times g||}$ if $x \times g \not= 0$ and $f_{k,\lambda}(x) = 0$ otherwise.

The rationale for the gain tuning is as follows: if the accelerometer measures a value $y$, we consider the smallest rotation giving to $g$ the same direction as $y$. Conserving the same axis, the angle of this rotation is thresholded (hence the parameter $\lambda$) to give less weight to outliers (without purely rejecting them, otherwise the filter couldn't converge when initialized too far). Then we choose as a gain function a rotation by a fraction of the obtained angle. We begin with the following preliminary result:
\begin{lm}\label{lemm}
For any $0 < k \leqslant 1$ and $0<\lambda \leqslant \pi$ the output error $\norm{Y_n-\hat R_n^Tg}$ associated to the observer \eqref{horizon} with noise turned off converges to 0.
\end{lm}
\begin{proof}
Let us consider the error evolution when the noise is turned off. It writes
$
\gamma_{n+1}=\gamma_n \exp(-f_{k,\lambda}(\gamma_n^{-1}g))
$, 
thus we have $
\gamma_{n+1}^{-1}g = \exp(f_{k,\lambda}(\gamma_n^{-1}g)) \gamma_n^{-1}g$. 
As for any $n \in \NN$, $f_{k,\lambda}(\gamma_n^{-1}g)$ is orthogonal to $g$ and  $\gamma_n^{-1}g$, $\gamma_{n+1}^{-1}g$ stays in the plane spanned by $g$ and $\gamma_0^{-1}g$, as well as the whole sequence $(\gamma_n^{-1} g)_{n \geqslant 0}$. Let $\phi_n=\text{angle}(\gamma_n^{-1}g,g)$. The dynamics of $(\phi_n)$ writes: $\phi_{n+1} = \phi_n - k.\min(\lambda, \phi_n)$. Thus $\phi_n$ goes to $0$, i.e: $\gamma_n^{-1}g \underset{n \rightarrow \infty}{\rightarrow} g$, i.e. the observation error goes to 0.
\end{proof}
The following result is a mere consequence of Lemma \ref{lemm} and Theorem \ref{thm::subgroup}.
\begin{prop}
\label{prop::K2}
The error variable associated to the filter defined by \eqref{horizon} converges to a stationary distribution for the T.V. norm, which does not depend on its initial distribution.
\end{prop}

\subsubsection{Numerical asymptotic gain optimization}To each couple $(k,\lambda)$ we can associate an asymptotic error dispersion (computed in the Lie Algebra) associated to the corresponding stationary distribution, and try to minimize it. As all computations are to be done off-line, the statistics of all distributions can be computed using particle methods. Table \ref{tab::horizon} gives the parameters of the model used in the following numerical experiment. Figure \ref{fig::Cuve} displays the Root Mean Square Error $RMSE = \sqrt{\mathbb{E}(\eta_\infty g - g)}$, computed over a grid for the parameters $(k,\lambda)$. The minimum is obtained for $k=0.1202$ and $\lambda=0.0029$.
If we compare it to a MEKF, we observe a huge difference. For our asymptotic invariant filter we get $RMSE=8.02 \times 10^{-4}$. For the MEKF, the observation noise matrix giving the best results leads to the value $RMSE=4.3 \times 10^{-3}$. This result is not surprising due to the fact that the outliers significantly pollute the estimates of the Kalman filter (see Fig. \ref{fig::Cuve}). This illustrates the fact that when the noise is highly non-Gaussian, an asymptotic gain with some optimality properties can still be found. 
\begin{table*}[h]
\centering
\caption{Artificial horizon: experiment parameters}
\begin{tabular}{|c|c|}
\hline
Standard Deviation of the model noise & $0.01$ degree $= 1.75 \times 10^{-4}$ rad \\
Standard Deviation of the regular observations & $0.1$ degree $= 1.75 \times 10^{-3}$ rad \\
Standard Deviation of the outliers & 30 degrees \\
Probability of the outliers to occur & 0.01 \\
\hline
\end{tabular}
\label{tab::horizon}
\end{table*}

%

\begin{figure}[!t]
\centering
\includegraphics[width=3in]{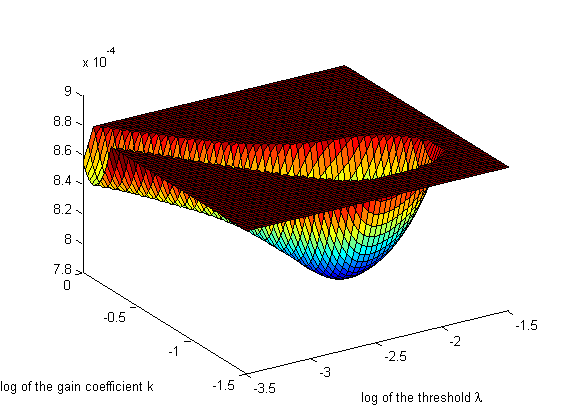}
\includegraphics[width=3in]{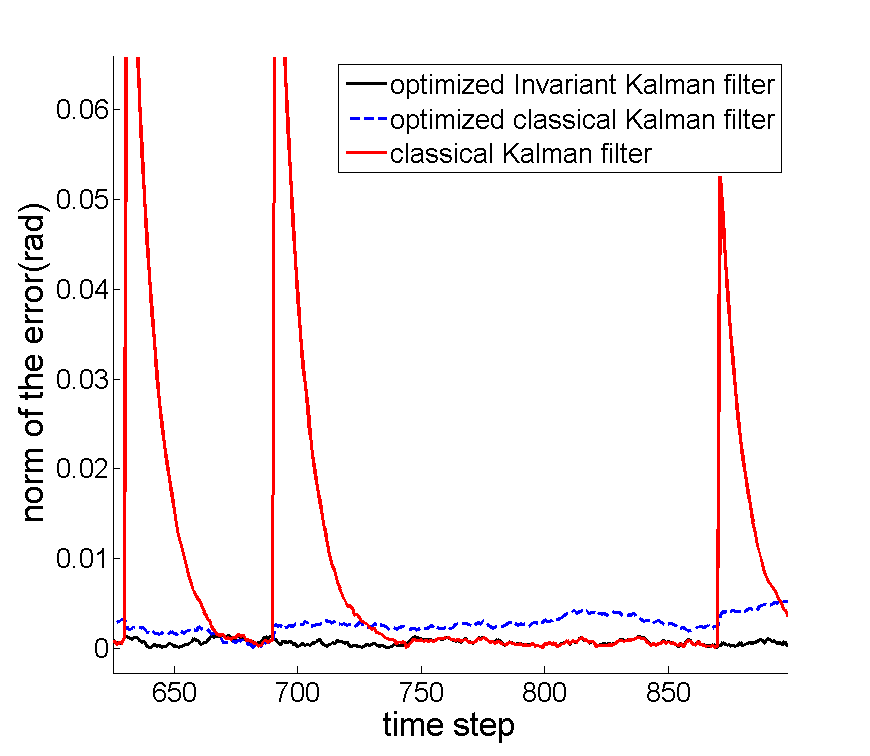} 
\caption{
Artificial horizon example. Top plot: off-line optimization of the parameters of the asymptotic invariant filter (axis $x \rightarrow \lambda$ ; axis $y \rightarrow k$ ; axis $z \rightarrow$ Root Mean Square Error). The highest values are not displayed to improve visualization. The performance is optimal for 
$k=0.1202, \lambda=0.0029$. Bottom plot: error evolution of three artificial horizons : the invariant filter with optimized gain functions (black), a MEKF ignoring outliers (red), and a MEKF where the observation noise covariance matrix has been numerically adjusted to minimize the RMSE (dashed line). We see that the MEKF cannot filter the outliers and at the same time be efficient in the absence of outliers, contrarily to the proposed invariant filter.
}
\label{fig::Cuve}
\end{figure}

\section{Gaussian filters}
\label{Gaussian_filters}
The present section focuses more on the transitory phase, and the gain is computed at each step based on Gaussian noise approximations and linearizations. Beyond the local optimization underlying the gain computation, the two following filters have the merit to readily convey an approximation of the extent of uncertainty carried by the estimations. We first introduce the discrete-time invariant extended Kalman filter (IEKF), which is an EKF, but based on the linearization of an invariant state error. Then we introduce the IEnKF, which computes the empirical covariance matrix of the error using particles. As all computations can be done off-line, it is easily implementable, as long as the gains can be stored, and is shown to convey a better indication of the error's dispersion for large noises. Simulations illustrate the results.

\subsection{The discrete-time Invariant EKF}\label{simus}

In the present section we introduce the discrete-time IEKF. Indeed, as in the standard EKF theory, the idea is to linearize the error equation \eqref{Markov1}-\eqref{Markov2} assuming the noises and the state error are small enough, use Kalman equations to tune the gains on this linear system, and implement the gain on the nonlinear model. As in the heuristic theory of the Invariant EKF in continuous time \cite{bonnabel2007left} \cite{bonnabel2009invariant} the gains and the error covariance matrix are proved to converge to a fixed value. In particular this allows (after some time) to advantageously replace the gain by its constant final value leading to a numerically very cheap asymptotic version of the IEKF.

\begin{algorithm}                      
\caption{Invariant Extended Kalman Filter}          
\label{algo::IEKF}                           
\begin{algorithmic}                    
\STATE{Returns at each time step $\hat{\chi}_n$, $P_n$ such that $\chi_n \approx \exp(\xi_n) \hat{\chi}_n$, where $\xi_n$ is Gaussian and $var (\xi_n)=P_n$. $\hat{\chi}_0$ and $P_0$ are inputs.}
\STATE{$H_\xi,H_V$ defined by:
\[ h(\exp({\xi'_n}),V_n)  = h(I_d,0) + H_\xi \xi'_n+ H_V V_n +O(\norm{\xi'_n}^2)+O(\norm{V_n}^2)\] }
\STATE{Set $n=0$.}
\LOOP
\STATE{Compute the value $M_{t_{n+1}}$ solving the equation \[M_{t_n}=0 , \qquad \frac{d}{dt}M_t=Var(w_t)+ad_{\upsilon_t} M_t ad_{\upsilon_t}^T\]}
\STATE{The process noise covariance $Q^w_n=Var(Ad_{\Upsilon_n}w_n)$ is equal to $M_{t_{n+1}}$ $\bigl($if $w_n$ is isotropic or if $\upsilon_t\equiv 0$, merely set $Q^w_n=Var(w_n)(t_{n+1}-t_{n})\bigr)$,}
\STATE{$Q^V_{n+1}=Var(V_{n+1})$,}
\STATE{$P_{n+1|n} = Ad_{\Upsilon_n}P_{n}Ad_{\Upsilon_n}^T + Q^w_n$,}
\STATE{$ S_{n+1} = H_V Q^V_{n+1} H_V^T + H_\xi P_{n+1|n} H_\xi^T$,}
\STATE{ $ L_{n+1} = P_{n+1|n} H_\xi^T S_{n+1}^{-1} $,}
\STATE{$P_{n+1} = (I-L_{n+1} H_\xi)P_{n+1|n} $,}
\STATE{ $\hat{\chi}'_{n+1} = \Upsilon_n \hat{\chi}_n \Omega_n$,}
\STATE{ $\hat{\chi}_{n+1} = \exp(L_{n+1} [ \hat{\chi'}_{n+1} Y_{n+1}-h(I_d,0) ] ) \hat{\chi'}_{n+1}$ }
\STATE{$n \leftarrow n+1$}

\ENDLOOP

\end{algorithmic}
\end{algorithm}

\subsubsection{Linearization of the equations and IEKF formulas}



\label{linearization}
We consider here the error equations \eqref{Markov1}-\eqref{Markov2}. Assuming errors and noises are small, we introduce their projection in the tangent space at the identity $I_d$  (the so-called Lie algebra $\mathfrak{g}$ of the group) using the matrix exponential map $\exp(.)$ (see e.g. \cite{olver-book99}). As the matrix vector space $\mathfrak{g}$ can be identified with $\RR^{\dim \mathfrak{g}}$ using the linear mapping $(.)_m$ (see Table  \ref{tab::formulaire}) we can assume $\exp: \RR^{\dim \mathfrak{g}} \rightarrow G$. This function is defined in Table \ref{tab::formulaire} for usual Lie groups of mechanics. Its inverse function will be denoted by $\exp^{-1}$. 
We thus define the following quantities of $\RR^{\dim \mathfrak{g}}$:
\begin{table*}[!t]
\renewcommand{\arraystretch}{1.3}
\caption{Useful formulas for some classical matrix Lie groups}
\label{tab::formulaire}
\centering
\begin{tabular}{|c||c||c|}
\hline
$G$ & $SO(3)$ ($\RR^{\dim \frak{g}}=\RR^3$) & ($\RR^{\dim \frak{g}}=\RR^6$) \\
\hline
\hline
Embedding of $G$ & $R \in \mathcal{M}_3(\RR), R^TR=I_3 $ & $\left( \begin{array}{cc}
R & v \\
0 & 1 \end{array} \right), R \in SO(3),u \in \RR^3 $ \\
Embedding of $\mathfrak{g}$ & $A_3: \psi \in \mathcal{M}_3(\RR), \psi^T=-\psi$ & $\left( \begin{array}{cc}
\psi & u \\
0 & 0 \end{array} \right), \psi \in A_3,u \in \RR^3 $ \\
$(.)_{m}$ & $\xi \rightarrow (\xi)_\times $ & $\left( \begin{array}{c}
\xi \\
 u \end{array} \right) \rightarrow \left( \begin{array}{cc}
(\xi)_\times & u \\
0 & 0 \end{array} \right)$\\
$x \rightarrow Ad_x$ & $R \rightarrow R$ & $\left( \begin{array}{c}
R \\
 T \end{array} \right) \rightarrow \left( \begin{array}{cc}
R & 0 \\
(T)_\times R & R \end{array} \right)$ \\

$\xi \rightarrow ad_\xi$ &  $\xi \rightarrow (\xi)_\times$ &  $\left( \begin{array}{c}
\xi \\
u \end{array} \right) \rightarrow \left( \begin{array}{cc}
(\xi)_\times & 0 \\
(u)_\times & (\xi)_\times \end{array} \right)$ \\
$\exp$ & $\exp(x)=I_3+\frac{sin(||x||)}{||x||}(x)_\times$ & $
\exp(x)=I_4+(x)_m+ \frac{[1-cos(||x||)]}{||x||^2}(x)_\times^2$ \\
 & $+ \frac{1}{||x||^2}(1-cos|x|)(x)_m^2$ & $+\frac{1}{||x||^3}(||x||-sin||x||)(x)_m^3$ \\
\hline
\end{tabular}
\end{table*}
\[
\xi_n=\exp^{-1}(\eta_n)
,\quad
\xi'_n=\exp^{-1}(\eta'_n)
, \quad
w_n=\exp^{-1}(W_n)
\]
The simplest way to design the gain is to use a  function which is linear in $\RR^{\dim \mathfrak{g}}$ and then mapped to $G$ through the exponential:
\begin{align*}
K_n: ~  y \rightarrow \exp[L_n( y-h(I_d,0))]
\end{align*}
Equations \eqref{Markov1} and \eqref{Markov2} mapped to $\RR^{\dim \mathfrak{g}}$ become: 
\begin{align*}
\xi'_{n+1} & = \exp^{-1}(\exp(Ad_{\Upsilon_n} w_n) \exp(Ad_{\Upsilon_n} \xi_n)), \\
\xi_{n+1} & = \exp^{-1}(\exp(\xi_{n+1}')\exp[-L_n( h(\exp(\xi'_{n+1},V_{n+1}))-h(I_d,0))]),
\end{align*}
where we have introduced the adjoint matrix $Ad$:
\begin{dfn}
\label{def::Ad}
The following property of the adjoint matrix $Ad$ can be used as a definition:
\[
\forall g \in G, \forall u \in \RR^{\dim \frak{g}}, ~ \exp(Ad_g u ) = g \exp(u) g^{-1}
\]

The differential of $\xi \rightarrow Ad_{\exp(\xi)}$ at zero is denoted $\xi \rightarrow ad_\xi$.

\end{dfn}
To tune the gains $L_n$ through the Kalman theory we need to evaluate at each time step, using $P_n = Var(\xi_n) \in \RR^{\dim \mathfrak{g} \times \dim \mathfrak{g}}$, the following quantities:
\[
P_{n+1|n} = Var({\xi'_{n+1}})
,\quad
S_{n+1} = Var(Y_{n+1})
\]
Our approach consists of using a coarse first-order development of the propagation and observation functions. To do so we use the Baker-Campbell-Hausdorff formula and keep only first order terms in $\xi'_{n+1}$, $w_n$ and $V_{n+1}$, the crossed terms being also neglected:
\begin{align*}
\xi'_{n+1} & = Ad_{\Upsilon_n} w_n + Ad_{\Upsilon_n} \xi_n \\
h(\exp(\xi'_{n+1}),V_{n+1}) & = h(I_d,0) + H_\xi \xi'_{n+1}+ H_V V_{n+1}\\
\xi_{n+1} & = \xi'_{n+1} - L_{n+1}(H_\xi \xi'_{n+1} + H_V V_{n+1})
\end{align*}
The computation of the gains is then straightforward using the Kalman theory in $\RR^{\dim \mathfrak{g}}$, which is a vector space. It is described in Algorithm \ref{algo::IEKF}.


\subsubsection{Convergence of the gains}
The main benefit of the filter is with respect to its convergence properties. Indeed, under very mild conditions, the covariance matrix $P_n$ and the filter's gain $K_n$ are proved to converge to fixed values \cite{Kalman-1961}. The practical consequences are at least twofold: 1-the error covariance converges to a fixed value and is thus much easier to interpret by the user than a matrix whose entries keep on changing (see Fig. \ref{fig::gains}) 2-due to computational limitations on-board, the covariance may be approximated by its asymptotic value, leading to an asymptotic version of the IEKF being numerically very cheap. 

\begin{thm}
If the noise matrices $Q^w_n=Var(w_n)$, $Q^V_n=Var(V_n)$ and the left inputs $\Upsilon_n$ are fixed, the pair $(Ad_{\Upsilon_n},H_\xi)$ is observable and $Q^w_n$ has full rank, then $P_n$ and $K_n$ converge to fixed values $P_\infty$ and $K_\infty$, as in the linear time-invariant case.
\end{thm}


\subsection{The discrete-time Invariant EnKF}
\label{sect::sampling}

When the error equation is independent of the inputs, the exact density of the error variable can be sampled off-line. This allows to  radically improve the precision of the quantities involved in the Kalman gains computation of Section \ref{linearization}. We propose here the Invariant Ensemble Kalman filter (IEnKF) described in Algorithm \ref{algo::MCIEKF}. The idea is to compute recursively through Monte-Carlo simulations a sampling of the error density and to use it, instead of linearizations, to evaluate precisely the innovation and error covariance matrices used to compute the gains in the Kalman filter. The procedure is described in Algorithm \ref{algo::MCIEKF}. 

\begin{algorithm}[!t]                      
\caption{IEnKF}          
\label{algo::MCIEKF}                           
\begin{algorithmic}                    
 	\STATE Define $H$ by $h(\exp(\xi),0)=h(I_d,0)+H\xi+O(\norm{\xi}^2)$
	\STATE Sample $M$ particles $(\eta_0^i)_{1 \leqslant i \leqslant M}$ following the prior error density
\FOR{$n=0$ \TO $N-1$}
	\FOR{$i=1$ \TO M}
		\STATE{ ${\eta'}_{n+1}^i=W_{n}^i \Upsilon \eta_{n}^i \Upsilon^{-1}$,}
		\STATE{$ y_{n+1}^i=h({\eta'}_{n+1}^i,V_{n+1}^i)$, }
	\ENDFOR
	\STATE{ $P_{n+1|n}=\frac{1}{M} \sum_{i=1}^M \exp^{-1}({\eta'}_{n+1}^i) \exp^{-1}({\eta'}_{n+1}^i)^T$}
	\STATE{ $S_{n+1}=\frac{1}{M} \sum_{i=1}^M y_{n+1}^i {y_{n+1}^i}^T$,}
	\STATE{$L_{n+1} = P_{n+1|n} H^T S_{n+1}^{-1}$ }	
	\FOR{$i=1$ \TO M}
\STATE{ $ {\eta}_{n+1}^i = {\eta'}_{n+1}^i \exp(-L_{n+1} [ y^i_{n+1}-h(I_d,0) ] )  $ }
	\ENDFOR
\STATE{Store the gain matrix $L_{n+1}$}
\ENDFOR
\STATE{The prior estimation $\hat{\chi}_0$ is an input.}
\FOR{$n=0$ \TO N-1}
\STATE{ $\hat{\chi}'_{n+1} = \Upsilon \hat{\chi}_{n} \Omega_{n}$,}
\STATE{$\hat{\chi}_{n+1} = \exp(L_{n+1} [ \hat{\chi'}_{n+1} Y_{n+1}-h(I_d,0) ] ) \hat{\chi'}_{n+1}$ }
\ENDFOR
\end{algorithmic}
\end{algorithm}

\subsection{Simulations}

\label{sect::simus_gaussian}

We will display here the results of Algorithms \ref{algo::IEKF} and  \ref{algo::MCIEKF} for the attitude estimation problem described in Section \ref{expl::2vect}:
\begin{equation*}
\begin{aligned}
\dotex{R}_t  & = R_t (\omega_t+w_t^\omega)_{\times}, \\
Y_n & = (R_{t_n}^T b_1+V_n^1, R_{t_n}^T b_2 + V_n^2),
\end{aligned}
\end{equation*}
where $R_t\in{SO(3)}$ represents the rotation that maps the body frame to the earth-fixed frame,  and $\omega_t$  is the instantaneous angular rotation vector which is assumed to be measured by noisy gyroscopes. $(Y_n)_{n\geq 0}$ is a sequence of discrete noisy measurements of the vectors $b_1$ and $b_2$ (for instance the gravity and the magnetic field), $V_n^1$ and $V_n^2$ are sequences of independent isotropic white noises. A simulation has been performed using the parameters given in Table \ref{tab::param2vect}. In this case the IEKF equations are described by Algorithm \ref{algo::IEKF_SO3}. As the state-of-the-art methods for this particular problem are the Multiplicative Extended Kalman Filter (MEKF, see e.g. \cite{Crassidis}) and the USQUE filter \cite{crassidis2003unscented} (a quaternion implementation of the Unscented Kalman Filter), the four methods are compared. The evolution of the gains is displayed on Fig. \ref{fig::gains} and shows the interest of the invariant approach. The error variable is expressed in the Lie algebra and its projection to the first axis is given on Fig. \ref{fig::2vect_MCIEKF} for each method (the projections on other axes being very similar, they are not displayed due to space limitations). All perform satisfactorily, but some differences are worthy of note. First we can see that the MEKF and the IEKF yield comparable performances but only the gains of the IEKF converge. Moreover, the linearizations lead the MEKF and the IEKF to fail capturing accurately the error dispersion through a  $3\sigma$-envelope. To this respect, the USQUE performs better than the two latter filters but still does not succeed in capturing the uncertainty very accurately. On the other hand, the envelope provided by the EnKF is very satisfying. This result is not surprising: this envelope has been computed using a sampling of the true density of the error, there is thus no reason why it should not be a valid approximation as long as there are sufficiently many particles. This method is thus to be preferred to the other ones when the user is willing and able to perform extensive simulations, and has the capacity to store the gains over the whole trajectory.

\begin{algorithm}                      
\caption{Invariant Extended Kalman Filter on $SO(3)$}          
\label{algo::IEKF_SO3}                           
\begin{algorithmic}                    
\STATE{Returns at each time step $\hat{R}_n$, $P_n$ such that $R_n \approx \exp(\xi_n) \hat{R}_n$, where $\xi_n$ is a centered Gaussian and $var (\xi_n)=P_n$.}
\STATE{$\hat{R}_0$ and $P_0$ are inputs.}
\STATE{$H_\xi = \left( \begin{array}{cc}
(b_1)_\times \\
(b_2)_\times 
 \end{array} \right), ~ H_V = I_6$,}
\STATE{$Q_V=\left( \begin{array}{cc}
Var(V^1) & 0 \\
 0 & Var(V^2) 
 \end{array} \right)$, $Q_w= \int_{t_n}^{t_{n+1}} Var(w_s^\omega) ds$.}

\STATE{Set $n=0$.}
\LOOP
\STATE{Let $\Omega_n$ be the solution at $t_{n+1}$ of $T_{t_n}=I_3$, $\frac{d}{dt} T_t = (\omega_t)_{\times}$.}
\STATE{$P_{n+1|n} = P_{n} + Q_w$,}
\STATE{$S_{n+1} = H_V Q_V H_V^T + H_\xi P_{n+1|n} H_\xi^T$ ,}
\STATE{ $ L_{n+1} = P_{n+1|n} H_\xi^T S_{n+1}^{-1}$,}
\STATE{$ P_{n+1} = (I-L_{n+1} H_\xi)P_{n+1|n} $,}
\STATE{ $\hat{R}'_{n+1} = \hat{R}_n \Omega_n$ }
\STATE{$\Delta_\chi=L_{n+1} [ \left( \begin{array}{c}
\hat{R'}_{n+1} Y_{n+1}^1 \\
 \hat{R'}_{n+1} Y_{n+1}^2
 \end{array} \right) -  \left( \begin{array}{c}
b_1 \\
b_2
 \end{array} \right)  ]$,}
\STATE{ $\hat{R}_{n+1} = \exp(\Delta_\chi) \hat{R'}_{n+1}$ (see Table \ref{tab::formulaire})}
\STATE{$n \leftarrow n+1$}
\ENDLOOP
\end{algorithmic}
\end{algorithm}

\begin{figure}[!t]
\centering
\begin{tabular}{cc}
\includegraphics[width=1.5 in]{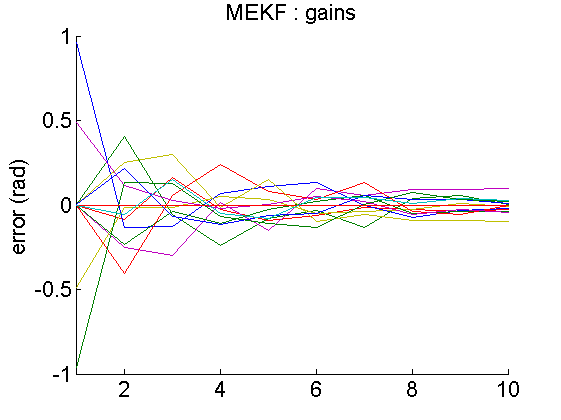} & \includegraphics[width=1.5 in]{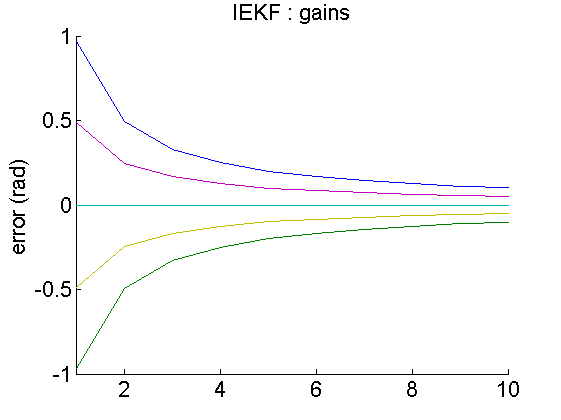} \\
\includegraphics[width=1.5 in]{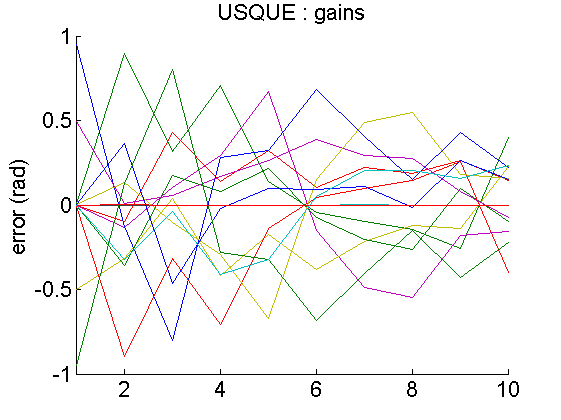} & \includegraphics[width=1.5 in]{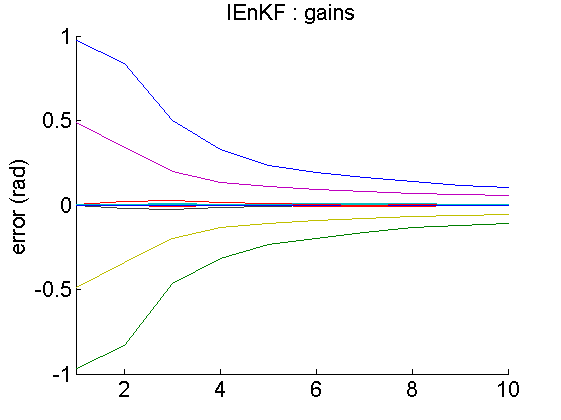}
\end{tabular}
\caption{ Evolution over time of the coefficients of the gain matrix for MEKF, IEKF, USQUE and IEnKF. The coefficients of the MEKF and USQUE have an erratic evolution, whereas those of the IEKF are convergent, allowing to save much computation power using asymptotic values. Moreover, in this case, only 4 of the 18 coefficients are not equal to zero for the IEKF, allowing sparse implementation.}
\label{fig::gains}
\end{figure}


\begin{table*}[!t]
\renewcommand{\arraystretch}{1.3}
\caption{Observation of two vectors: parameters}
\label{tab::param2vect}
\centering
\begin{tabular}{|c||c|c|c|c|}
\hline
\bfseries Parameter & 
$b_1$ &
 $b_2$ & 
$Q^{V_1}$ &
$Q^{V_2}$ \\
\hline
\bfseries Value &
 $(1,0,0)$ &
$(0,1,0)$  &
$0,0873^2 I_3$ &
$0,0873^2 I_3$\\
\hline
\hline
\bfseries Parameter &
$P_0$ &
$Q^w$ &
$N$ &
Simulations \\
\hline
\bfseries Value &
$0.5236^2 I_3$ &
$0,01745^2 I_3$ &
$50$ &
$1000$\\
\hline
\end{tabular}
\end{table*}
\begin{figure}[!t]
\centering
\begin{tabular}{cc}
\includegraphics[width=1.5 in]{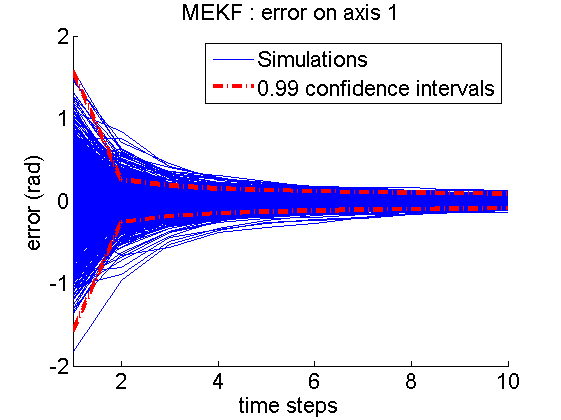} & \includegraphics[width=1.5 in]{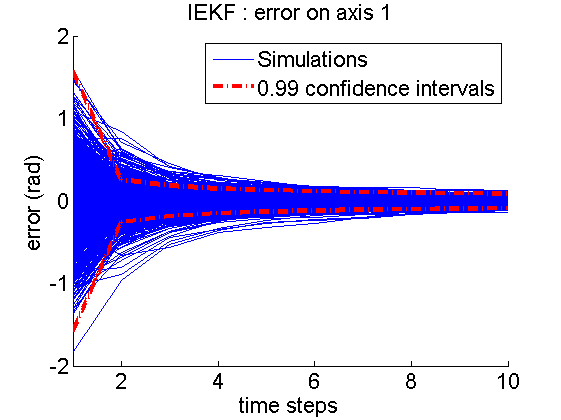} \\ \includegraphics[width=1.5 in]{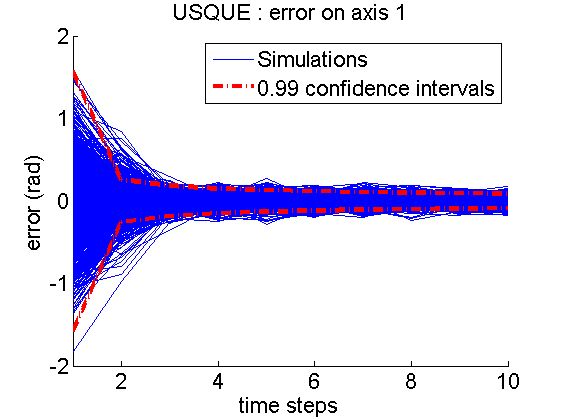} & \includegraphics[width=1.5 in]{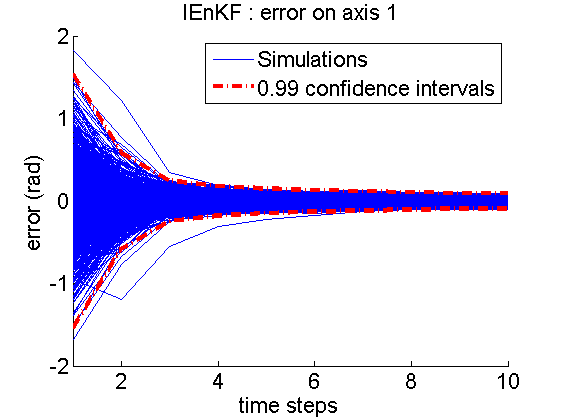}
\end{tabular}
\caption{ Evolution of the error variable, projected on the first axis, for 1000 simulations using MEKF, IEKF, USQUE and IEnKF. The MEKF and the IEKF give similar results. The USQUE is doing slightly better but is outperformed by the IEnKF. Only the latter captures the error dispersion satisfactorily, as the $3 \sigma$ envelope contains 99 $\%$ of the simulated trajectories.}
\label{fig::2vect_MCIEKF}
\end{figure}

\section{Conclusion}

This paper has introduced a proper stochastic framework for a class of filtering problems on Lie groups. A wide class of discrete-time filters with strong properties has been proposed, along with several methods to tune the gains. When the gain is held fixed, the estimation error has the remarkable property to provably converge in distribution. Another route is to focus more on the transitory phase, and to seek to minimize the dispersion of the estimation error at each step based on linearizations and Gaussian noise approximations. This has lead to an invariant version of the celebrated EKF, having the property that the error covariance matrix converges asymptotically, as well as to an invariant version of the EnKF which has been shown to convey an accurate indication of the level of uncertainty carried by the estimate.

The results of the paper open the door to a rich variety of filters on Lie groups. Indeed, all the classical gain tuning approaches can still be considered, and a method based on the unscented transform for instance, completing the work \cite{crassidis2003unscented} seems an interesting idea. In a situation where many observations are available at each time step, an information filter could also be derived mimicking the IEKF. Besides, we believe that the cases where the gains of the IEKF are sparse should be given special attention, by studying when this property occurs and if significant computation power and memory can be saved. The possibilities offered by off-line simulations deserve also to be deepened. The computation of exact confidence sets can be explored, establishing a link with the approach of set-valued observers proposed in \cite{bras2013global} and \cite{sanyal2008global}. Concerning the asymptotic method introduced here, the optimization has been performed quite naively as efficient density simulation is beyond the scope of the paper, but there may certainly be some room for improvements using Markov Chain Monte-Carlo (MCMC) schemes to speed up convergence. In the future we will concentrate on proving the convergence in distribution of the IEKF, that can be conjectured at this stage, but does not seem easy.


\appendix

\section{Proof of Proposition \ref{isomorphism}}

For any $i\in \NN$, let $\mathcal{M}_i(\RR)$ denote the set of square matrices of size $i$. Consider the isomorphisms $\phi_X: X \in \RR^N \rightarrow \begin{pmatrix}
Id_N & 0_{N,1} \\
X^T & 1
 \end{pmatrix} \in \mathcal{M}_{N+1}(\RR)$  and $\phi_Y: Y \in \RR^p \rightarrow \begin{pmatrix}
H^T \\
Y^T
\end{pmatrix} \in \mathcal{M}_{N+1,p}$.
Letting $\chi_n = \phi_X(X_n)$, $Y_n^*=\phi_Y(Y_n)$, $V_n^*=\phi_Y(V_n)$, $W_n^*=\phi_X(W_n)$, $\Omega_n=\phi_X(B_n)$ and $\Upsilon_n=Id_{N+1}$ equations \eqref{separation_obs} become $ \chi_{n+1} = \Upsilon_n W_n^* \chi_n \Omega_n$ and $Y_n^* = \chi_n^{-1}V^*_n$. Using the inverse isomorphisms $\phi_X^{-1}$ and $\phi_Y^{-1}$, the invariant update $\hat{\chi}_{n+1}=K_n(\hat{\chi}'_{n+1} Y^*_{n+1})\hat{\chi}_{n+1}$ corresponds in $\RR^n$ to \eqref{separation_obs}.

\section{Proofs of the results of Section \ref{cv_results}}

We first introduce here the basic notions about Harris Chains that we will need, to prove the stochastic convergence properties of the invariant filters.
\begin{dfn}{(First hitting time)}
Let $ S $ be a measurable space and $(X_n)_{n \geqslant 0}$ a Markov chain taking values in $S$. The first hitting time $\tau_C$ of a subset $C \subset S$ is defined by $\tau_C = \inf_{n \geqslant 0} ~ \{ X_n \in C \}$.
\end{dfn}
\begin{dfn}{(Recurrent aperiodic Harris chain)}
\label{def::Harris}
Let $ S $ be a measurable space and $(X_n)_{n \geqslant 0}$ a Markov chain taking values in $S$. $(X_n)_{n \geqslant 0}$ is said to be a recurrent aperiodic Harris chain if there exist two sets $A,B \subset S$ satisfying the following properties:
\begin{enumerate}[i]
\item For any initial state $X_0$ the first hitting time of $A$ is a.s. finite. \label{hitting_time}
\item 
There exists a probability measure $\rho$ on $B$, and  $\epsilon >0$ such that if $x \in A$ and $D \subset B$ then $\mathbb{P}(X_1 \in D | X_0 =x) > \epsilon \rho(D)$. \label{mixing}
\item There exists an integer $N \geqslant 0$ such that: $\forall x \in A, \forall n \geqslant N, \mathbb{P}(X_n \in A | X_0=x)>0$. \label{aperiodic}
\end{enumerate}
\end{dfn}
The  somewhat  technical property \ref{mixing} means that any given area of $B$ can be reached from each point of $A$ with non-vanishing probability.
\begin{thm}\text{[Harris, 1956]}
\label{thm::Harris}
A recurrent aperiodic Harris chain admits a unique stationary distribution $\rho_\infty$ and the density of the state $X_n$ converges to $\rho_\infty$ in T.V. norm for any distribution of $X_0$.
\end{thm}
\begin{proof}
See \cite{durrett2010probability}, Theorems 6.5 and 6.8
\end{proof}
In other words, if for any initialization we can ensure that the process will come back with probability $1$ to a central area $A$ where mixing occurs, we have a convergence property. The following technical result is a cornerstone to demonstrate the theorems of Subsection \ref{cv_results}.
\begin{lm}
\label{MainProp}
Let $G$ be a locally compact Lie group provided with a right-invariant distance $d$, and $C \subset G$ be a measurable compact set endowed with the $\sigma$-algebra $\Sigma$. Consider a  homogeneous Markov chain $(\eta_n)_{n \geqslant 0}$ defined by the relation
$$
\eta_0\in C,\qquad \eta_{n+1}=q_z(\eta_n)
$$where $z \in \mathcal Z$ is some random variable belonging to a measurable space. Let $0$ denote a specific point of $\mathcal Z$.  Let  $Q: (C,\Sigma) \rightarrow \mathbb{R}_+$, denote the transition kernel  of the chain, that is $Q(x, V)=\mathbb P(q_z(\eta_n)\in V\mid \eta_n=x)$. We assume that  
\begin{enumerate}[(a)]
\item there exist real numbers $\alpha,\epsilon >0$ such that for any $x \in C$ and $V \subset  \mathcal{B}_o(q_0(x),\alpha)$, the latter denoting the open ball center $q_0(x)$ and radius $\alpha$, we have $Q(x, V) =\mathbb P(q_z(\eta_n)\in V\mid \eta_n=x) \geqslant \epsilon \mu_G(V)$ where $\mu_G$ is the right-invariant Haar measure. \label{main_diffusion}
\item $q_0$ admits a fixed point $x_0 \in C$, i.e. $q_0(x_0) = x_0$. \label{main_fixed_point}
\end{enumerate}
Let $U_0 = \{ x_0 \}$. Define the sets $(U_n)_{n \geqslant 0}$   recursively by $U_n' =  \{ x \in G, d(x,U_n) < \frac{\alpha}{2}  \} \cap C $ and $U_{n+1}= q_0^{-1}(U_n' )$. If there exists an integer $N \geqslant 0$ such that $U_N = C$, then the p.d.f. of $\eta_n$ converges for the T.V. norm and its limit does not depend on the initialization.
\end{lm}
More prosaically, $z$ denotes the cumulated effect of the model and observation noise, and it is assumed that 1- the noiseless algorithm has some global convergence properties and 2- there are sufficiently many small enough noises so that 
there is an $\epsilon$ chance for $\eta_n$ to jump from $x$ to any neighborhood of the noiseless iterate $q_0(x)$ with a multiplicative factor corresponding to the neighborhood's size. It then defines a sequence of sets, by picking a fixed point of $q_0$, dilating it, and taking its pre-image. Dilatation and inversion are then reiterated until the whole set $C$ is covered. In this case, forgetting of the initial distribution is ensured.

\begin{proof}
We will demonstrate the property trough three intermediate results:
\begin{enumerate}
\item There exists a sequence $\epsilon_1,\cdots,\epsilon_N$ such that: $\mathbb P(\eta_1 \in U_n | \eta_0 \in U_{n+1})>\epsilon_n >0$. \label{property1}
\item The first hitting time of $U_1$ is a.s. finite. \label{property2}
\item $(\eta_n)_{n\geq 0}$ is an aperiodic recurrent Harris chain with $A=U_1$ and $B=U_0'$. \label{property3}
\end{enumerate}

The conclusion will then immediately follow from Theorem \ref{thm::Harris}. We give first a qualitative explanation of the approach. We have by assumption a sequence of sets $(U_n)_{n \geqslant 0}$. Intermediate result \ref{property1}) states that starting from a set $U_{n+1}$ the chain has a non-vanishing chance to jump to the ``smaller" set $U_n$ at the next time step. Intermediate result \ref{property2}) states that starting from $U_{n+1}$, because of the non-vanishing property \ref{property1}), the chain cannot avoid $U_n$ forever and eventually reaches $U_n$, $U_{n-1}$ and finally $U_1$. Intermediate result \ref{property3}) brings out the fact that, $A=U_0'$ being totally included in the ball of radius $\alpha$ centered on any of its points, the chain can go from any given point of its pre-image $B=U_1$ to any area of $A$ with controlled probability, which is the last property needed to obtain a recurrent aperiodic Harris chain.

\begin{enumerate}
\item  Consider the closure $\overline{U_n'}$ of $U_n'$. For any $x \in \overline{U_n'}$ the set $\mathcal{B}_o(x,\alpha) \cap U_n$ is open and non-empty (as $d(x,U_n) \leqslant \frac{\alpha}{2} $ by definition of $\overline{U_n'}$) so $\mu_G(\mathcal{B}_o(x,\alpha) \cap U_n) >0$. The function $f : x \rightarrow \mu_G(\mathcal{B}_o(x,\alpha) \cap U_n) >0$ is continuous (since $\mu_G$ is a regular and positive measure) on the compact set $\overline{U_n'}$, so it admits a minimum $m_n > 0$. We get: $P(\eta_1 \in U_n | \eta_0 \in U_{n+1} )\geq P(\eta_1 \in \mathcal{B}_o (q_0(\eta_0),\alpha) \cap U_n | \eta_0 \in U_{n+1}) = P(\eta_1 \in \mathcal{B}_o (q_0(\eta_0),\alpha) \cap U_n | q_0(\eta_0) \in U_n') \geq \epsilon m_n $ using the fact that, by assumption (\ref{main_diffusion}), $Q(x,V)\geq\epsilon\mu_G(V)$ for $x \in C$ and $V \subset \mathcal{B}_o(q_0(x),\alpha)$. We set $\epsilon_n = \epsilon m_n$ to prove intermediate result \ref{property1}). 


\item We will prove by descending induction on $n$ the property $\mathcal{P}_n: P(\tau_n<\infty)=1$. $\mathcal{P}_N$ is obvious as $U_N = C$. Now assume $\mathcal{P}_{n+1}$ is true for $n \geqslant 1$. Due to homogeneity, we can construct a strictly increasing sequence of stopping times $(\nu_p)_{p \geqslant 0}$ such that: $\forall p \geqslant 0, \eta_{\nu_p} \in U_{n+1}$. But   $  \mathbb{P}(\eta_{\nu_p+1}\in U_n | \eta_{\nu_p} \in U_{n+1}) > \epsilon_n $ by intermediate result \ref{property1}). Thus letting $T_k$ denote the event  $\{ \forall i \leqslant k, \eta_i \notin  U_n \} = \{ k < \tau_n \}$, we have for any $n \geqslant 0$:
 $
\mathbb{P}(T_{\nu_p+1})  = \mathbb{P}(T_{\nu_p}) \mathbb{P}(T_{\nu_p+1}|T_{\nu_p})  \leqslant \mathbb{P}(T_{\nu_p}) (1-\epsilon_n)
$. 
Thus 
$\mathbb{P}(T_{\nu_{p+1}}) \leqslant \mathbb{P}(T_{\nu_p +1}) \leqslant \mathbb{P}(T_{\nu_p})(1-\epsilon_n)$. A quick induction and a standard application of Borel-Cantelli Lemma prove there exists a.s. a rank $p$ such that $T_{\nu_p}$ is false. In other words, $\mathcal{P}_n$ is true. In particular, $\mathcal{P}_1$ is true: the first hitting time of $U_1$ is a.s. finite. 

\item Define $A=U_1$ and  $B=U_0'= \mathcal{B}_o(x_0,\frac{\alpha}{2}) \cap C$ . For any $V \subset B$ and $x \subset A$ we have $V \subset \mathcal{B}_o(q_0(x),\alpha)$, thus $Q(x,V) \geqslant \epsilon \mu_G(V)$ by assumption (\ref{main_diffusion}). Thus $(\eta_n)_{n \geqslant 0}$ verifies property \ref{mixing}) of Definition \ref{def::Harris}. As $A=U_1$ the intermediate result \ref{property2}) shows that $(\eta_n)_{n \geqslant 0}$ verifies the property \ref{hitting_time}) of Definition \ref{def::Harris}. As there exists a.s. a rank $k$ such that $\eta_k \in U_1$ by intermediate result \ref{property2}), there exists an integer $M$ such that $\mathbb{P}(\eta_M \in U_1)>0$. As $U_1 \cap U_0' \subset \mathcal{B}_o(q(x),\alpha)$ for any $x \in U_1$ (by definition of $U_0'$ and $U_1$) Assumption (\ref{main_diffusion}) gives immediately $\mathbb{P}(\eta_{M+1} \in U_1 \cap U'_0) > \mathbb{P}(\eta_{M} \in U_1) \epsilon \mu_G(U_1 \cap U'_0)$. Using only the definition of $U_0'$ and $U_1$, and Assumption (\ref{main_fixed_point}), we see easily that both contain $x_0$. Thus the set $U_0' \cap U_1$ is non-empty and open.  We have then $\mu_G(U_1 \cap U'_0)>0$ and $\mathbb{P}(\eta_{M+1} \in U_1 \cap U'_0) > 0$. By an obvious induction on $m$,  $\mathbb{P}(\eta_{m} \in U_1 \cap U'_0)>0$ for any $m>M$. As $U_1 \cap U'_0 \subset A$ we get $\mathbb{P}(\eta_m \in A)>0$ for any $m>M$ and the property 
iii) of Definition \ref{def::Harris} is also verified. We finally obtain that $(\eta_n)_{n \geqslant 0}$ is an aperiodic recurrent Harris chain.
\end{enumerate}
Using Theorem \ref{thm::Harris} we can conclude that the density of $\eta_n$ converges in T.V. norm to its unique equilibrium distribution for any initialization.
\end{proof}

Building upon the previous results, the proof of Theorem \ref{thm::global} is as follows. The Markov chain defined by equations \eqref{Markov1} and \eqref{Markov2} with fixed $\Upsilon$ can be written under the form $\eta'_{n+1}=q_z(\eta'_n)=q_{1,z} \circ q_{2,z}(\eta'_n)$ with $z=(W_n,V_n) \in G \times \RR^p $, $q_{1,z}(x)=\Upsilon W_n x \Upsilon^{-1}$ and $q_{2,z}(x)=xK(h(x,V_n))^{-1}$. Let $0=(I_d,0)$. We will show that $(\eta'_n)_{n \geqslant 0}$ has all the properties required in Lemma \ref{MainProp}.

Let $\alpha, \epsilon, \epsilon'$ be defined as in \ref{cv_results}. Let $\alpha'$ such as $\mathcal{B}_o(I_d, \alpha') \subset \Upsilon \mathcal{B}_o(I_d,\frac{\alpha}{2}) \Upsilon^{-1}$ and $\epsilon''=\epsilon' |Ad_ {\Upsilon^{-1}}| \epsilon$, where $|Ad_{\Upsilon^{-1}}|$ is the determinant of the adjoint operator on $\mathfrak{g}$. For any $x \in C$ and $V \subset \mathcal{B}_o(q_0(x),\alpha')=\mathcal{B}_o(I_d,\alpha')q_0(x)$ we have:
\begin{align*}
P(q_z(x) \in V) &  \geqslant P(q_z(x) \in V,q_{2,z}(x) \in \mathcal{B}_o (q_{2,0}(x), \frac{\alpha}{2})) \\
 &  \geqslant P(q_{2,z}(x) \in \mathcal{B}_o (q_{2,0}(x), \frac{\alpha}{2})) \times \\
 &  P(\Upsilon W_n q_{2,z}(x) \Upsilon^{-1} \in V|q_{2,z}(x) \in \mathcal{B}_o (q_{2,0}(x), \frac{\alpha}{2}))
\end{align*}
i.e. $ P(q_z(x) \in V) \geqslant \epsilon' P(W_n \in \Upsilon^{-1} V \Upsilon q_{2,z}(x)^{-1} |q_{2,z}(x) \in \mathcal{B}_o (q_{2,0}(x), \frac{\alpha}{2}))$.
Assuming $q_{2,z}(x) \in \mathcal{B}_o (q_{2,0}(x),\frac{\alpha}{2})$ we have:
$ \Upsilon^{-1} V \Upsilon q_{2,z}(x)^{-1} \subset \Upsilon^{-1} \mathcal{B}_o(I_d,\alpha') q_0(x) \Upsilon q_{2,z}(x)^{-1} \subset \mathcal{B}_o(I_d, \frac{\alpha}{2}) \Upsilon^{-1} q_0(x) \Upsilon q_{2,z}(x)^{-1}  \subset \mathcal{B}_o(I_d,\frac{\alpha}{2}) q_{2,0}(x) q_{2,z}(x)^{-1}  \subset \mathcal{B}_o(I_d,\frac{\alpha}{2}+\frac{\alpha}{2})
$.
As $W_n$ is independent from the other variables we obtain:
$
P(W_n \in \Upsilon^{-1} V \Upsilon q_{2,z}(x)^{-1} |q_{2,z}(x) \in \mathcal{B}_o (q_{2,0}(x), \frac{\alpha}{2})) \geqslant \epsilon \mu_G(\Upsilon^{-1} V \Upsilon q_{2,z}(x)^{-1}) = \epsilon \mu_G(\Upsilon^{-1} V \Upsilon) = \epsilon \mu_G(V) |Ad_{\Upsilon^{-1}}|
$
And finally:
\[
P(q_z(x) \in V) \geqslant \epsilon' |Ad_{\Upsilon^{-1}}| \epsilon \mu(V) = \epsilon''\mu(V)
\]

 As $q_0$ has $I_d$ as a fix point  we only have to verify that the sets $U_n$ as defined in Lemma \ref{MainProp} eventually cover the whole set $C$. It suffices to consider for any $n \geqslant 0$ the set $D_n = \{ x \in C, \forall k \geqslant n, q_0^k(x) \in U'_0 \}$. As we have $q_0^n(x) \rightarrow I_d$ almost-everywhere on $C$ we get: $\mu_G( \cup_{n \geqslant 0} D_n) = \mu_G(C)$. As the sequence of sets $(D_n)_{n \geqslant 0}$ increases we have: $\mu_G(D_n) \underset{n \rightarrow \infty}{\longrightarrow} \mu_G(C)$. We introduce here the quantity $v_{\min}=\min_{x \in C} \mu_G(\mathcal{B}_o(x, \frac{\alpha'}{2}) \cap C )$ (the property $C=cl(C^o)$ and the regularity of $\mu_G$ ensure that we have $v_{\min}>0$). Let $N \in \mathbb{N}$ be such that $\mu_G(D_N) > \mu_G(C)-v_{\min}$. We have then: $\forall y \in C, d(y, D_N)<\frac{\alpha'}{2}$ (otherwise we would have $\mu_G(D_N) \leqslant \mu_G(C)-v_{\min}$). As we have $D_N \subset U_N$ we obtain $\forall y \in C, d(y, D_N)<\frac{\alpha'}{2}$ thus $U'_N=C$ and $U_{N+1}=q_0^{-1}(U_N')=q_0^{-1}(C)=C$. So all the conditions of Lemma \ref{MainProp} hold and we can conclude  convergence in T.V. norm to a stationary distribution which doesn't depend on the prior, provided that its support lies in $C$. Theorem \ref{thm::global} is proved.

Corollary \ref{thm::global2} follows directly from the case $C=G$.

Theorem \ref{thm::subgroup} can be proved as follows. Let $C=G$. Let $K = \{ x \in G, h(x,0)=h(I_d,0)  \}$. Note that the left-right equivariance assumption ensures that $K$ is a subgroup of $G$. \ First we show that there exists an integer $N_1$ such that $K \subset U_{N_1}$. As we have $\forall x \in K$ $q_0(x)=\Upsilon x \Upsilon^{-1}$ the sequence of sets $Q_n=\Upsilon^{-n} U_n \Upsilon^n \cap K$ is growing and we have: $\{ x \in K, ~d(x,Q_n)<\frac{\alpha}{2} \} \subset Q_{n+1}$.The set $Q_{\infty}= \cup^\infty _{n=1} Q_n$ is open in $K$ as an union of open sets, but we have $\forall x \in Q_\infty,\mathcal{B}_o(x,\frac{\alpha}{2}) \cap K \subset Q_\infty $ thus $\forall x \in K \setminus Q_\infty,\mathcal{B}_o(x,\frac{\alpha}{4}) \cap K \subset K \setminus Q_\infty$. This implies that $Q_\infty$ and $K \setminus Q_\infty$ are both open in $K$. As $K$ is connected (see assumptions of Theorem \ref{thm::subgroup}) we obtain $Q_\infty = K$. As $K$ is compact (as a closed subset of a compact) the open cover $\cup _{i=1}^\infty Q_n$ has a finite subcover and there exists an integer $N_1$ such that $Q_{N_1}=K$, i.e. $\Upsilon^{N_1} K \Upsilon^{-N_1} = U_{N_1}$. As $K$ is a subgroup of $G$ (see above) containing $\Upsilon$ we obtain $K \subset U_{N_1}$. Now we have to prove that the sets $(U_n)_{n \geqslant 0}$ eventually cover the whole set $C$. For any $x \in G$ we have $h(q_0^n(x),0) \rightarrow h(I_d,0)$. Thus there exists a rank $n$ such that $\forall k \geqslant n, q_0^k(x) \in U_N'$ (otherwise we could extract a subsequence from $(q_0^n(x))_{n \geqslant 0}$ which stays at a distance $> \frac{\alpha}{2}$ from $K$, as $G$ is compact we could extract again a convergent subsequence and its limit $q_0^\infty (x)$ would be outside $K$, thus we would have $h(q_0^n(x)) \not\rightarrow h(I_d,0)$). Here we can define as in the proof of Theorem \ref{thm::global} the sets $D_n = \{ x \in C, \forall k \geqslant n, q_0^k (x) \in U_{N_1}'  \}$ (note that $U_1'$ has been replaced by $U_{N_1}'$ in this definition of $D_n$). As for almost any $x\in G$ there exists a rank $n$ such that $x \in D_n$ we have $\mu_G(D_n) \rightarrow \mu_G(C)$ and as in the proof of Theorem \ref{thm::global} the sets $U_n$ eventually cover the set $C$.

The second part of the property (invariance to left multiplication) is easier. Consider the error process $\zeta_n = \tilde{\Upsilon} \eta_n$. The propagation and update steps read:
\begin{align*}
\zeta_{n+1}'=\tilde{\Upsilon} \eta_{n+1}'=\tilde{\Upsilon} \Upsilon W_n \eta_n \Upsilon^{-1}= \Upsilon \tilde{\Upsilon} W_n \eta_n \Upsilon^{-1} & \overset{\mathcal{L}}{=}\Upsilon W_n \tilde{\Upsilon} \eta_n \Upsilon^{-1} \\
 & \overset{\mathcal{L}}{=}\Upsilon W_n \zeta \Upsilon^{-1}
\end{align*}
\begin{align*}
\zeta_{n+1} & =\tilde{\Upsilon} \eta_{n+1}' K_{n+1}(h(\eta_{n+1}', V_{n+1}))^{-1}  \\
 & \overset{\mathcal{L}}{=} \tilde{\Upsilon} \eta_{n+1}' K_{n+1}(h(\tilde{\Upsilon} \eta_{n+1}', V_{n+1}))^{-1} \\
 & \overset{\mathcal{L}}{=} \zeta_{n+1}' K_{n+1}(h(\zeta', V_{n+1}))^{-1}
\end{align*}
where we have used the property: $h(\eta_{n+1}', V_{n+1}) \overset{\mathcal{L}}{=} \eta_{n+1}'^{-1} h(I_d, V_{n+1}) \overset{\mathcal{L}}{=} \eta_{n+1}'^{-1} h(\tilde{\Upsilon}, V_{n+1}) \overset{\mathcal{L}}{=} h(\tilde{\Upsilon} \eta_{n+1}', V_{n+1})$.
We see that the law of the error process is invariant under left multiplication by $\tilde{\Upsilon}$, thus the asymptotic distribution $\pi$ inherits this property.

\section{Proofs of the results of Section \ref{asymptotic2vect}}

\subsection{Proof of Proposition \ref{thm::discrete}}

We define the continuous process $\tilde{\gamma}_t: \RR \rightarrow G$ as follows:
\[
\forall n \in \mathbb{N}, \tilde{\gamma}_n= \gamma_n
\]
\[
\forall n\in \mathbb{N}, \forall t \in [n,n+1[, \dotex{\tilde{\gamma}}_t = -\tilde{\gamma}_t k(\tilde{\gamma}_n)
\]
We obtain immediately that $\tilde{\gamma}_t$ is continuous. Besides, for any $n \in \mathbb{N}$ and $t \in ]n,n+1[$ one has:
\begin{align*}
|\frac{d}{dt} \langle k(\tilde{\gamma}_t), k(\tilde{\gamma}_n) \rangle| & = |\langle \frac{d}{dt} k(\tilde{\gamma}_t), k(\tilde{\gamma}_n) \rangle| \\ 
 &  \leqslant | \frac{d}{dt} k(\tilde{\gamma}_t) | | k(\tilde{\gamma}_n) |  \\
 & \leqslant | \frac{d}{dt}\tilde{\gamma}_t | | k(\tilde{\gamma}_n) |  \text{ (due to $|\frac{\partial k}{\partial x}| \leqslant 1$ )}\\
 & \leqslant | k(\tilde{\gamma}_n) |^2 \text{ (as $\frac{d}{dt}\tilde{\gamma}_t = -\tilde{\gamma}_t k(\tilde{\gamma}_n)$ )}
\end{align*}
Thus:
\[
\langle k(\tilde{\gamma}_t), k(\tilde{\gamma}_n) \rangle \geqslant \langle k(\tilde{\gamma}_n), k(\tilde{\gamma}_n) \rangle -(t-n)|k(\tilde{\gamma}_n)|^2 \geqslant 0
\]
proving in the Lie algebra that $
\langle k(\tilde{\gamma}_t), k(\tilde{\gamma}_n) \rangle \geqslant 0$. 
As $u$ is a positive function we have thus $
u(\tilde{\gamma}_t)^{-1} \langle \tilde{\gamma}_t k(\tilde{\gamma}_t), -\tilde{\gamma}_t k(\tilde{\gamma}_n) \rangle \leqslant 0
$ immediately proving 
$\langle grad_E(\tilde{\gamma_t}) , \frac{d}{dt}\tilde{\gamma}_t \rangle \leqslant 0$. The latter result being true for every $t \in \RR \backslash \mathbb{N}$ and the function $E(\tilde{\gamma}_t)$ being continuous, it decreases and thus converges on $\mathbb{R}_{\geq 0}$. As we have supposed the sublevel sets of $E$ are bounded (and closed as $E$ is continuous), $\tilde{\gamma}_t$ is stuck in a compact set. Let $\gamma_\infty$ an adherence value of $\tilde{\gamma}_n$. By continuity of $E$ we have $E(\tilde{\gamma}_\infty) = E(\tilde{\gamma}_\infty \exp(-k(\tilde{\gamma}_\infty))$. Thus the function $t \rightarrow E(\tilde{\gamma}_{\infty}\exp(-tk(\tilde{\gamma}_\infty)))$ is decreasing on $[0,1]$ and has the same value at $0$ and $1$. Thus it is constant, its derivative at 0 is null proving $\gamma_\infty=I_d$. $(\tilde{\gamma_n})_{n \geqslant 0}$ being confined in a compact set and having $I_d$ as unique adherence value we finally get $\gamma_n = \tilde{\gamma_n} \rightarrow I_d$.


\subsection{Proof of Proposition \ref{prop::K1}}

Let $\gamma_n = R_n \hat{R}_n^T$ be the invariant error. Its equation reduces to:
\[
\gamma_{n+1} = \gamma_{n}.\exp(-k_1(\gamma_n^T b_1) \times b_1-k_2(\gamma_n^T b_2) \times b_2)
\]
Let $k(\gamma)=k_1(\gamma^T b_1) \times b_1+k_2(\gamma^T b_2) \times b_2$ and $E: \gamma \rightarrow k_1 ||\gamma^T b_1-b_1||^2 + k_2 ||\gamma^T b_2 - b_2||^2$. To apply Proposition \ref{thm::discrete} we will first verify that 1) $\forall \gamma \in SO(3), \gamma . k(\gamma) = grad_E(\gamma)$
2) $|\frac{\partial k}{\partial \gamma}| \leqslant 1 $. 
To prove 1) consider the dynamics in $\in SO(3)$ defined by $\frac{d}{dt}\gamma_t = \gamma_t (\psi_t)_{\times}$ for some rotation vector path $\psi_t$. We have 
$\frac{d}{dt}E(\gamma_t)  = k_1(\gamma_t^Tb_1-b_1)^T \frac{d}{dt}\gamma_t^Tb_1 + k_2(\gamma_t^Tb_2-b_2)^T \frac{d}{dt}\gamma_t^Tb_2$. Using triple product equalities this is equal to 
 $\langle k_1 (\gamma_t^Tb_1)_\times b_1 +  k_2 (\gamma_t^Tb_2)_\times b_2, \psi_t \rangle= \langle k(\gamma_t), \psi_t \rangle  = \langle\gamma_t. k(\gamma_t), \frac{d}{dt}\gamma_t \rangle$.
Thus $\gamma.k(\gamma)=grad_E(\gamma)$. To prove 2) we analogously see that $
\frac{d}{dt}k(\gamma_t) = -[k_1(b_1)_\times (\gamma_t^Tb_1)_\times + k_2(b_2)_\times (\gamma_t^Tb_2)_\times] \psi_t$. Thus $|\frac{d}{dt}k(\gamma_t)|  \leqslant (k_1+k_2)|\psi_t|$
and finally $|\frac{\partial k}{\partial \gamma}| \leqslant 1 $. Now, except if initially $\gamma_0$ is the rotation of axis $b_1\times b_2$ and angle $\pi$, the function $E$ is strictly decreasing, and $I_d$ is the only point in the sublevel set $\{\gamma\in SO(3)\mid E(\gamma)\leq E(\gamma_0)\}$ such that $grad_E(\gamma)=0$. Applying Proposition \ref{thm::discrete} allows to prove Proposition \ref{prop::K1}.


\bibliographystyle{plain}
\bibliography{rhn}



\end{document}